%% file: main.tex
\documentclass{article}

\usepackage{arxiv}
\usepackage[utf8]{inputenc} 
\usepackage[T1]{fontenc}    
\usepackage{hyperref}       
\usepackage{amsfonts}       
\usepackage{nicefrac}       
\usepackage{microtype}      
\usepackage{graphicx}
\usepackage[square,numbers]{natbib}

\usepackage{enumitem}
\usepackage{amsmath}
\usepackage{amsthm}
\usepackage{algorithm}
\usepackage{algpseudocode}
\usepackage[caption=false]{subfig}

\newtheorem{theorem}{Theorem}[section]
\newtheorem{proposition}[theorem]{Proposition}
\newtheorem{lemma}[theorem]{Lemma}

\theoremstyle{definition}
\newtheorem{definition}[theorem]{Definition}
\newtheorem{assumption}[theorem]{Assumption}
\theoremstyle{remark}
\newtheorem{remark}{Remark}

\DeclareMathOperator*{\argmin}{argmin}

\def\0{\boldsymbol{0}}
\def\g{\boldsymbol{g}}
\def\u{\boldsymbol{u}}
\def\v{\boldsymbol{v}}
\def\x{\boldsymbol{x}}
\def\y{\boldsymbol{y}}
\def\z{\boldsymbol{z}}
\def\A{\boldsymbol{A}}
\def\B{\mathcal{B}}
\def\N{\mathbb{N}}
\def\R{\mathbb{R}}
\def\W{\boldsymbol{W}}
\def\Z{\mathbb{Z}}

\title{Scalable Distributed Optimization of Multi-Dimensional Functions Despite Byzantine Adversaries}
\date{} 

\author{ \href{https://orcid.org/0000-0002-0769-1399}{\includegraphics[scale=0.06]{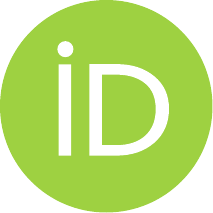}\hspace{1mm}Kananart~Kuwaranancharoen} \\
	Intel Labs\\
	Intel Corporation\\
	Hillsboro, OR 97124 \\
        \href{mailto:kananart.kuwaranancharoen@intel.com}{\texttt{kananart.kuwaranancharoen@intel.com}} \\
	\And
        \href{https://orcid.org/0000-0001-6349-5774}{\includegraphics[scale=0.06]{orcid.pdf}\hspace{1mm}Lei~Xin} \\
        Department of Computer Science and Engineering \\
	The Chinese University of Hong Kong\\
        Ma Liu Shui, Hong Kong \\
	\href{mailto:lxinshenqing@gmail.com}{\texttt{lxinshenqing@gmail.com}} \\
        \And
	\href{https://orcid.org/0000-0002-5390-2505}{\includegraphics[scale=0.06]{orcid.pdf}\hspace{1mm}Shreyas~Sundaram} \\
	School of Electrical and Computer Engineering\\
	Purdue University\\
	West Lafayette, IN 47907 \\
	\href{mailto:sundara2@purdue.edu}{\texttt{sundara2@purdue.edu}} \\
}

\renewcommand{\shorttitle}{Scalable Distributed Optimization Despite Byzantine Adversaries}

\hypersetup{
pdftitle={Scalable Distributed Optimization Despite Byzantine Adversaries},
pdfauthor={Kananart~Kuwaranancharoen, Lei~Xin, Shreyas~Sundaram},
pdfkeywords={Byzantine Attacks, Convex Optimization, Distributed Algorithms, Fault Tolerant Systems, Graph Theory, Machine Learning, Multi-Agent Systems, Network Security},
}

\begin{document}
\maketitle

\begin{abstract}
The problem of distributed optimization requires a group of networked agents to compute a parameter that minimizes the average of their local cost functions. While there are a variety of distributed optimization algorithms that can solve this problem, they are typically vulnerable to ``Byzantine'' agents that do not follow the algorithm.  Recent attempts to address this issue focus on single dimensional functions, or assume certain statistical properties of the functions at the agents. In this paper, we provide two resilient, scalable, distributed optimization algorithms for multi-dimensional functions. Our schemes involve two filters, (1) a distance-based filter and (2) a min-max filter, which each remove neighborhood states that are extreme (defined precisely in our algorithms) at each iteration. We show that these algorithms can mitigate the impact of up to $F$ (unknown) Byzantine agents in the neighborhood of each regular agent.  In particular, we show that if the network topology satisfies certain conditions, all of the regular agents' states are guaranteed to converge to a bounded region that contains the minimizer of the average of the regular agents' functions.
\end{abstract}


\input{contents/sec-intro}
\input{contents/sec-notation_problem}
\input{contents/sec-algorithm}
\input{contents/sec-result}
\input{contents/sec-discussion}
\input{contents/sec-experiment}
\input{contents/sec-conclusion}

\bibliographystyle{unsrtnat}
\bibliography{references}  

\newpage
\appendix
\input{contents/sec-supp_lemma}
\input{contents/sec-supp_proof_aux}
\input{contents/sec-supp_proof_grad}
\input{contents/sec-supp_proof_true}
\input{contents/sec-supp_proof_conv}

\end{document}

%% file: contents/sec-intro.tex
\section{Introduction} 
\label{sec: introduction}

The design of distributed algorithms has received significant attention in the past few decades \cite{tsitsiklis1986distributed, xiao2006optimal}. In particular, for the problem of distributed optimization, a set of agents in a network are required to reach agreement on a parameter that minimizes the average of their local objective functions, using information received from their neighbors \cite{wang2011control,boyd2011distributed,eisen2017decentralized,xin2019frost}. A variety of approaches have been proposed to tackle different challenges of this problem, e.g., distributed optimization under constraints \cite{zhu2011distributed}, distributed optimization under time-varying graphs \cite{nedic2014distributed}, and distributed optimization for non-convex non-smooth functions \cite{zeng2018nonconvex}.  However, these existing works typically make the assumption that all agents are trustworthy and cooperative (i.e., they follow the prescribed protocol); indeed, such protocols fail if even a single agent behaves in a malicious or incorrect manner \cite{sundaram2018distributed}.

As security becomes a more important consideration in large scale systems, it is crucial to develop algorithms that are resilient to agents that do not follow the prescribed algorithm. A handful of recent papers have considered fault tolerant algorithms for the case where agent misbehavior follows specific patterns \cite{ravi2019case,wu2018data}. A more general (and serious) form of misbehavior is captured by the {\it Byzantine} adversary model from computer science, where misbehaving agents can send arbitrary (and conflicting) values to their neighbors at each iteration of the algorithm. Under such Byzantine behavior, it has been shown that it is impossible to guarantee computation of the true optimal point \cite{sundaram2018distributed,su2015byzantine}. Thus, researchers have begun formulating distributed optimization algorithms that allow the non-adversarial nodes to converge to a certain region surrounding the true minimizer, regardless of the adversaries' actions \cite{su2015byzantine, su2016fault, sundaram2018distributed, zhao2019resilient}.

It is worth noting that one major limitation of the above works \cite{su2015byzantine, su2016fault, sundaram2018distributed, zhao2019resilient} is that they all make the assumption of scalar-valued objective functions, and the extension of the above ideas to general multi-dimensional convex functions remains largely open. In fact, one major challenge for minimizing multi-dimensional functions is that the region containing the minimizer of the sum of functions is itself difficult to characterize. Specifically, in contrast to the case of scalar functions, where the global minimizer\footnote{We will use the terms ``global minimizer" and ``minimizer of the sum" interchangeably since we only consider convex functions.} always lies within the smallest interval containing all local minimizers, the region containing the minimizer of the sum of multi-dimensional functions may not necessarily be in the convex hull of the minimizers \cite{kuwaran2018location}. 

There exists a branch of literature focusing on secure distributed machine learning in a client-server architecture \cite{gupta2019byzantine, blanchard2017machine, pillutla2022robust}, where the server appropriately filters the information received from the clients. However, their extensions to a distributed (peer-to-peer) setting remains unclear. The papers \cite{yang2019byrdie, fang2022bridge, elkordy2022basil} consider a vector version of the resilient machine learning problem in a distributed (peer-to-peer) setting. These papers show that the states of regular nodes will converge to the statistical minimizer with high probability (as the amount of data of each node goes to infinity), but the analysis is restricted to i.i.d training data across the network. However, when each agent has a finite amount of data, these algorithms are still vulnerable to sophisticated attacks as shown in \cite{guo2020towards}. The work \cite{gupta2021byzantine} considers a Byzantine distributed optimization problem for multi-dimensional functions, but relies on redundancy among the local functions, and also requires the underlying communication network to be complete. The work presented in \cite{ravi2019detection} proposes a resilient algorithm under statistical characteristic assumptions but lacks guarantees. The recent work \cite{wu2023byzantine} studies resilient stochastic optimization problem under non-convex and smooth assumptions on local functions, which differs from our focus. The algorithm proposed in that work achieves convergence to a stationary point up to a constant error but does not ensure asymptotic consensus. Additionally, the recent work \cite{kuwaran2023geometric} offers convergence guarantees to a neighborhood of the optimal solution under deterministic settings, but it pertains to a distinct class of functions -- strongly convex and smooth functions.

To the best of our knowledge, our conference paper \cite{kuwaran2020byzantine} is the first one that provides a scalable algorithm with convergence guarantees in general networks under very general conditions on the multi-dimensional convex functions held by the agents in the presence of Byzantine faults. Different from existing works, the algorithm in \cite{kuwaran2020byzantine} does not rely on any statistical assumptions or redundancy of local functions. Technically, the analysis addresses the challenge of finding a region that contains the global minimizer for multiple-dimensional functions, and shows that regular states are guaranteed to converge to that region under the proposed algorithm. The Distance-MinMax Filtering Dynamics in \cite{kuwaran2020byzantine} requires each regular node to compute an auxiliary point using resilient asymptotic consensus techniques on their individual functions' minimizers in advance. After that, there are two filtering steps in the main algorithm that help regular nodes to discard extreme states. The first step is to remove extreme states (based on the distance to the auxiliary point), and the second step is to remove states that have extreme values in any of their components. On the other hand, the algorithm in \cite{kuwaran2020byzantine} suffers from the need to compute the auxiliary point prior to running the main algorithm, since the fixed auxiliary point is only achieved by the resilient consensus algorithm asymptotically.

In this paper, we eliminate this drawback. The algorithms and analysis we propose here expand upon the work in \cite{kuwaran2020byzantine} in the following significant ways. First, the algorithms in this paper bring the computation of the auxiliary point into the main algorithm, so that the local update of auxiliary point and local filtering strategies are performed simultaneously. This makes the analysis much more involved since we need to take into account the coupled dynamics of the estimated auxiliary point and the optimization variables. Second, the algorithms make better use of local information by including each regular node's own state as a metric. In practice, we observe that this performs better than the approach in \cite{kuwaran2020byzantine}, since each agent may discard fewer states and hence, there are more non-extreme states that can help the regular agents get close to the true global minimizer. Again, we characterize the convergence region that all regular states are guaranteed to converge to using the proposed algorithm. Third, we present an alternate algorithm in this paper which only makes use of the distance filter (as opposed to both the distance and min-max filter); we show that this algorithm significantly reduces the requirements on the network topology for our convergence guarantees, at the cost of losing guarantees on consensus of the regular nodes' states. Importantly, our work represents the first attempt to provide convergence guarantees in a geometric sense, characterizing a region where all states are ensured to converge to, without relying on any statistical assumptions or redundancy of local functions.

Our paper is organized as follows. Section~\ref{sec: notation and problem} introduces various mathematical preliminaries, and states the problem of resilient distributed optimization. We provide our proposed algorithms in Section~\ref{sec: algorithm}. We then state the assumptions and some important results related to properties of the proposed algorithms in Section~\ref{sec: assumption result}. In Section~\ref{sec: discussion}, we provide discussion on the results. Finally, we simulate our algorithms to numerically evaluate their performance in Section~\ref{sec: simulations}, and conclude in Section~\ref{sec: conclusion}.

%% file: contents/sec-notation_problem.tex
\section{Mathematical Notation and Problem Formulation} 
\label{sec: notation and problem}

Let $\N$, $\Z$ and $\R$ denote the set of natural numbers (including zero), integers, and real numbers, respectively. We also denote the set of positive integers by $\Z_+$. The cardinality of a set is denoted by $|\cdot|$. The set of subgradients of a convex function $f$ at point $\x$ is called the subdifferential of $f$ at $\x$, and is denoted $\partial f( \x)$.

\subsection{Linear Algebra}
Vectors are taken to be column vectors, unless otherwise noted.  We use $x^{(\ell)}$ to represent the $\ell$-th component of a vector $\x$.
The Euclidean norm on $\mathbb R^d$ is denoted by $\| \cdot \|$. We denote by $\langle \u, \v \rangle$ the Euclidean inner product of $\u$ and $\v$, i.e., $\langle \u, \v \rangle = \u^T \v$ and by $\angle (\u, \v)$ the angle between vectors $\u$ and $\v$, i.e., $\angle (\u, \v) = \arccos \big( \frac{ \langle \u, \v \rangle }{\Vert \u \Vert \Vert \v \Vert}  \big)$. We use $\mathcal{S}_d^+$ to denote the set of positive definite matrices in $\R^{d \times d}$. The Euclidean ball in $d$-dimensional space with center at $\x_0$ and radius $r \in \R_{> 0}$ is denoted by $\B(\x_0, r) := \{ \x \in \R^d : \| \x - \x_0 \| \leq r \} $.

\subsection{Graph Theory}
We denote a network by a directed graph $\mathcal G = ( \mathcal{V}, \mathcal{E})$, which consists of the set of nodes  $\mathcal{V} = \{ v_{1}, v_{2}, \ldots, v_{N} \}$ and the set of edges $\mathcal{E} \subseteq \mathcal{V} \times \mathcal{V}$. If $(v_i, v_j) \in \mathcal{E}$, then node $v_j$ can receive information from node $v_i$. 
The in-neighbor and out-neighbor sets are denoted by $\mathcal{N}^{\text{in}}_{i} = \{ v_j \in \mathcal V: \; (v_{j}, v_{i}) \in \mathcal{E} \} $ and $\mathcal{N}^{\text{out}}_{i} = \{ v_j \in \mathcal V: \; (v_{i}, v_{j}) \in \mathcal{E} \} $, respectively. A path from node $v_{i}\in \mathcal V$ to node $v_{j}\in \mathcal V$  is a sequence of nodes $v_{k_1},v_{k_2}, \dots, v_{k_l}$ such that $v_{k_1}=v_i$, $v_{k_l}=v_j$ and $(v_{k_r},v_{k_{r+1}}) \in \mathcal{E}$ for $1\leq r \leq l-1$. Throughout the paper, the terms nodes and agents will be used interchangeably. Given a set of vectors $\{\x_1, \x_2, \ldots, \x_N\}$, where each $\x_i \in \R^d$ , we define for all $\mathcal{S} \subseteq \mathcal{V}$,
\begin{equation*}
    \{ \x_i \}_{\mathcal{S}} := \{ \x_i \in \R^d : v_i \in \mathcal{S} \}.
\end{equation*}

\begin{definition}
A graph $\mathcal G = ( \mathcal{V}, \mathcal{E})$ is
said to be rooted at node $v_i \in \mathcal{V}$ if for all nodes $v_j \in \mathcal{V} \setminus \{v_i\}$, there is a path from $v_i$ to $v_j$. A graph is said to be rooted if it is rooted at some node $v_{i}\in \mathcal V$.
\end{definition}

We will rely on the following definitions from \cite{leblanc2013resilient}. 

\begin{definition}[$r$-reachable set]
For a given graph $\mathcal{G}$ and a positive integer $r \in \Z_+$, a subset of nodes $\mathcal S  \subseteq \mathcal V$ is said to be $r$-reachable if there exists a node $v_i \in \mathcal S$  such that $|\mathcal N^{\textup{in}}_{i} \setminus \mathcal S| \geq r$.
\end{definition}

\begin{definition}[$r$-robust graph] 
For $r \in \mathbb Z_+$, a graph $\mathcal G$ is said to be $r$-robust if for all pairs of disjoint nonempty subsets $\mathcal{S}_1, \mathcal{S}_2 \subset \mathcal V$, at least one of $\mathcal{S}_1$ or $\mathcal{S}_2$ is $r$-reachable.
\end{definition}

The above definitions capture the idea that sets of nodes should contain individual nodes that have a sufficient number of neighbors outside that set. This will be important for the {\it local} decisions made by each node in the network under our algorithm, and will allow information from the rest of the network to penetrate into different sets of nodes.

\subsection{Adversarial Behavior}
\begin{definition}
A node $v_i \in \mathcal V$ is said to be Byzantine if during each iteration of the prescribed algorithm, it is capable of sending arbitrary (and perhaps conflicting) values to different neighbors. It is also allowed to update its local information arbitrarily at each iteration of any prescribed algorithm.
\end{definition}

The set of Byzantine nodes is denoted by $\mathcal A \subset \mathcal V$. The set of regular nodes is denoted by  $\mathcal{R} = \mathcal{V}\setminus\mathcal{A}$.

The identities of the Byzantine agents are unknown to regular agents in advance. Furthermore, we allow the Byzantine agents to know the entire topology of the network, functions equipped by the regular nodes, and the deployed algorithm. In addition, Byzantine agents are allowed to coordinate with other Byzantine agents and access the current and previous information contained by the nodes in the network (e.g. current and previous states of all nodes). Such extreme behavior is typical in the study of the adversarial models \cite{sundaram2018distributed, su2015byzantine,yang2019byrdie}. In exchange for allowing such extreme behavior, we will consider a limitation on the number of such adversaries in the neighborhood of each regular node, as follows. 

\begin{definition}[$F$-local model]
For $F \in \mathbb Z_+$, we say that the set of adversaries $\mathcal{A}$ is an $F$-local set if $|\mathcal{N}^{\textup{in}}_{i} \cap \mathcal{A} | \leq F$, for all $v_{i} \in \mathcal R$.
\end{definition}

Thus, the $F$-local model captures the idea that each regular node has at most $F$ Byzantine in-neighbors.

\subsection{Problem Formulation}

Consider a group of $N$ agents $\mathcal{V}$ interconnected over a graph $\mathcal{G} = (\mathcal{V}, \mathcal{E})$. Each agent $v_i \in \mathcal{V}$ has a local convex cost function $f_i: \R^d \rightarrow \R$.  The objective is to collaboratively solve the minimization problem
\begin{align} 
    \mathop{\min}_{\x \in \R^d} \frac{1}{N}\sum _{v_i \in \mathcal{V}} f_{i}(\x), \label{prob: min_all_f}
\end{align}
where $\x \in \R^d$ is the common decision variable. 
A common approach to solve such problems is for each agent to maintain a local estimate of the solution to the above problem, which it iteratively updates based on communications with its immediate neighbors.
However, since Byzantine nodes are allowed to send arbitrary values to their neighbors at each iteration of any algorithm, it is not possible to solve  Problem \eqref{prob: min_all_f} under such misbehavior (since one is not guaranteed to infer any information about the true functions of the Byzantine agents) \cite{su2015byzantine, sundaram2018distributed}.  Thus, the optimization problem is recast into the following form:
\begin{align} 
    \mathop{\min}_{\x \in \R^d} \frac{1}{|\mathcal R|}\sum _{v_i \in \mathcal R } f_{i}(\x),  \label{prob: regular node} 
\end{align}
i.e., we restrict our attention only to the functions held by regular nodes. 

\begin{remark}
In the resilient distributed optimization problem, the agents are required to compute a value that (approximately) minimizes the sum of functions held by each (regular) agent.  Compared to the resilient consensus problem, this necessitates more information than simply the initial vectors held by each agent (even if those vectors are initialized to be the local minimizers of the agents' functions).  Indeed, the need to combine estimates of the multi-dimensional minimizer from neighbors, while incorporating gradient dynamics, all in a resilient fashion is what makes the resilient distributed optimization problem more difficult than the standard consensus problem.
\end{remark}

\begin{remark}  \label{rem: fund limit}
The additional challenge in solving the above problem lies in the fact that no regular agent is aware of the identities or actions of the Byzantine agents. Furthermore, in the worst-case scenario, it is not feasible to achieve an exact solution to Problem~\ref{prob: regular node}, as the Byzantine agents can modify the functions while still adhering to the algorithm, making it impossible to differentiate them \cite{su2015byzantine, sundaram2018distributed}.
\end{remark}

In the next section, we propose two scalable algorithms that allow the regular nodes to approximately solve the above problem, regardless of the identities or actions of the Byzantine agents (as proven later in the paper). 

%% file: contents/sec-algorithm.tex
\section{Resilient Distributed Optimization Algorithms} 
\label{sec: algorithm}

\subsection{Proposed Algorithms}
The algorithms that we propose are stated as Algorithm~\ref{alg: DMM_filter} and Algorithm~\ref{alg: D_filter}. We start with Algorithm~1. At each time-step $k$, each regular node\footnote{Byzantine nodes do not necessarily need to follow the above algorithm, and can update their states however they wish.} $v_i \in \mathcal{R}$ maintains and updates a vector $\x_i[k] \in \R^d$, which is its estimate of the solution to Problem~\eqref{prob: regular node}, and a vector $\y_i[k] \in \R^d$, which is its estimate of an auxiliary point that provides a general sense of direction for each agent to follow. 

\begin{remark}
The purpose of the estimates $\x_i [k]$ is to be an approximation to the minimizer of the sum of the functions.  To update this estimate, the agents have to decide which of the estimates provided by their neighbors to retain at each iteration of the algorithm (since up to $F$ of those neighboring estimates may be adversarially chosen by Byzantine agents).  To help each regular agent decide which estimates to keep, the auxiliary points $\y_i [k]$ are used to perform the distance-based filtering step (\textbf{Line~7}). In fact, each auxiliary point provides a general sense of direction for the agents' estimates, and thus helps them filter out adversarial estimates that attempt to draw them away from the true minimizer.
\end{remark}

We now explain each step used in  Algorithm \ref{alg: DMM_filter} in detail.\footnote{In the algorithm, $\mathcal{X}_i[k]$, $\mathcal{X}^{\text{dist}}_i[k]$, $\mathcal{X}^{\text{mm}}_i[k]$, $\mathcal{Y}_i[k]$ and $\mathcal{Y}^{\text{mm}}_i[k]$ are multisets.}

\begin{algorithm}
\caption{Simultaneous Distance-MinMax Filtering Dynamics } \label{alg: DMM_filter}
\textbf{Input} Network $\mathcal{G}$, 
functions $\{f_i\}_{i = 1}^{N}$, parameter $F$
\begin{algorithmic}[1]
\State Each $v_i \in \mathcal{R}$ sets $\hat{\x}_i^* \gets$ \texttt{optimize}($f_i$)   
\State Each $v_i \in \mathcal{R}$ sets $\x_i[0] \gets \hat{\x}_i^*$ and $\y_i[0] \gets \hat{\x}_i^*$
\For {$k = 0, 1, 2, 3, \ldots$}  
\For {$v_i \in \mathcal R$}  \Comment{Implement in parallel}

\quad \textbf{Step I:} Broadcast and Receive
\State \texttt{broadcast}($\mathcal{N}_i^\text{out}$, $\x_i[k]$, $\y_i[k]$)
\State $\mathcal{X}_i[k], \; \mathcal{Y}_i[k] \gets$ \texttt{receive}($\mathcal{N}_i^\text{in}$)

\quad \textbf{Step II:} Resilient Consensus Step
\State $\mathcal{X}_i^{\text{dist}}[k] \gets$ \texttt{dist\_filt}($F, \; \y_i[k], \; \mathcal{X}_i[k]$)
\State $\mathcal{X}_i^{\text{mm}}[k] \gets$ \texttt{x\_minmax\_filt}($F ,\; \mathcal{X}_i^{\text{dist}}[k]$)
\State $\z_i[k] \gets$ \texttt{x\_weighted\_average}($ \mathcal{X}_i^{\text{mm}}[k]$)

\quad \textbf{Step III:} Gradient Update
\State $\x_i[k+1] \gets$ \texttt{gradient}($f_i, \z_i[k]$)

\quad \textbf{Step IV:} Update the Estimated Auxiliary Point
\State $\mathcal{Y}_i^{\text{mm}}[k] \gets$ \texttt{y\_minmax\_filt}($F, \; \mathcal{Y}_i[k]$)
\State $\y_i[k+1] \gets$ \texttt{y\_weighted\_average}($ \mathcal{Y}_i^{\text{mm}}[k]$)
\EndFor
\EndFor
\end{algorithmic}
\end{algorithm}

\begin{itemize}
    \item \textbf{Line 1:} $\hat{\x}_i^* \gets$ \texttt{optimize} ($f_i$)  \\
    Each node $v_{i} \in \mathcal R$ uses any appropriate optimization algorithm to get an approximate minimizer $\hat{\x}_i^* \in \R^d$ of its local function $f_i$. We assume that there exists $\epsilon^* \in \R_{\geq 0}$ such that the algorithm achieves $\| \hat{\x}_i^* - \x_i^* \| \leq \epsilon^*$
    for all $v_i \in \mathcal{R}$ where $\x_i^* \in \R^d$ is a true minimizer of the function $f_i$; we assume formally that such a true (but not necessary unique) minimizer exists for each $v_i \in \mathcal{R}$ in the next section.
    
    \item \textbf{Line 2:} $\x_i[0] \gets \hat{\x}_i^*$ and $\y_i[0] \gets \hat{\x}_i^*$ \\
    Each node $v_{i} \in \mathcal R$ initializes its own estimated solution to Problem \eqref{prob: regular node}  ($\x_i[0] \in \R^d$) and estimated auxiliary point ($\y_i[0] \in \R^d$) to be $\hat{\x}_i^*$.
    
    \item \textbf{Line 5:} \texttt{broadcast} ($\mathcal{N}_i^{\text{out}}$, $\x_i[k]$, $\y_i[k]$) \\
    Node $v_{i} \in \mathcal{R}$ broadcasts its current state $\x_{i}[k]$ and estimated auxiliary point $\y_{i}[k]$ to its out-neighbors $\mathcal{N}_i^{\text{out}}$.
    
    \item \textbf{Line 6:} $\mathcal{X}_i[k], \; \mathcal{Y}_i[k] \gets$ \texttt{receive}($\mathcal{N}_i^{\text{in}}$) \\
    Node $v_{i} \in \mathcal{R}$ receives the current states $\x_j[k]$ and $\y_j[k]$ from its in-neighbors $\mathcal{N}_i^{\text{in}}$. 
    So, at time step $k$, node $v_i$ possesses the sets of states\footnote{In case a regular node $v_i$ has a Byzantine neighbor $v_j$, we abuse notation and take the value $\x_j [k]$ to be the value received from node $v_j$ (i.e., it does not have to represent the true state of node $v_j$).}
    \begin{equation*}
         \mathcal{X}_i[k] := \big\{ \x_j[k] \in \R^d: v_j \in \mathcal{N}_i^{\text{in}} \cup \{ v_i \} \big\} 
         \quad \text{and} \quad
         \mathcal{Y}_i[k] := \big\{ \y_j[k] \in \R^d: v_j \in \mathcal{N}_i^{\text{in}} \cup \{ v_i \} \big\}.
    \end{equation*}
    The sets $\mathcal{X}_i[k]$ and $\mathcal{Y}_i[k]$ have an indirect relationship through the distance-based filter (\textbf{Line~7}) as only $\y_i [k] \in \mathcal{Y}_i[k]$ is used as the reference to remove states in $\mathcal{X}_i [k]$.
     
    \item \textbf{Line 7:} $\mathcal{X}_i^{\text{dist}}[k] \gets$ \texttt{dist\_filt}($F, \; \y_i[k], \; \mathcal{X}_i[k]$) \\
    Intuitively, regular node $v_i$ ignores the states that are far away from its own auxiliary state $\y_i[k]$ in $L^2$ sense. Formally, node $v_i \in \mathcal R$ computes the distance between each vector in $\mathcal{X}_i[k]$ and its own estimated auxiliary point $\y_i[k]$:
    \begin{align}
        \mathcal{D}_{ij}[k] := \| \x_j[k] - \y_i[k] \|  \;\; \text{for} \;\; \x_j[k] \in \mathcal{X}_i[k].    \label{def: D-metric}
    \end{align}
    Then, node $v_i \in \mathcal{R}$ sorts the values in the set $\{ \mathcal{D}_{ij} [k] : v_j \in \mathcal{N}_i^{\text{in}} \cup \{ v_i \} \}$ and removes the $F$ largest values that are larger than its own value $\mathcal{D}_{ii} [k]$. If there are fewer than $F$ values higher than its own value, $v_i$ removes all of those values. Ties in values are broken arbitrarily. The corresponding states of the remaining values are stored in $\mathcal{X}_i^{\text{dist}}[k]$. 
    In other words, regular node $v_i$ removes up to $F$ of its neighbors' vectors that are furthest away from the auxiliary point $\y_i[k]$.
    
    \item \textbf{Line 8:} $\mathcal{X}_i^{\text{mm}}[k] \gets$ \texttt{x\_minmax\_filt}($F ,\; \mathcal{X}_i^{\text{dist}}[k]$) \\
    Intuitively, regular node $v_i$ ignores the states that contains extreme values in any of their components in the ordering sense. Formally, for each time-step $k \in \N$ and dimension $\ell \in \{ 1, 2, \ldots, d \}$, 
    define the set $\mathcal{V}_i^{\text{remove}} (\ell) [k] \subseteq \mathcal{N}_i^{\text{in}}$, where a node $v_j$ is in $\mathcal{V}_i^{\text{remove}} (\ell) [k]$ if and only if
    \begin{itemize}
        \item $x_j^{(\ell)} [k]$ is within the $F$-largest values of $\big\{ x_r^{(\ell)} [k] \in \R: \x_r [k] \in \mathcal{X}^{\text{dist}}_i [k] \big\}$ and $x_j^{(\ell)} [k] > x_i^{(\ell)} [k]$, or
        \item $x_j^{(\ell)} [k]$ is within the $F$-smallest values of $\big\{ x_r^{(\ell)} [k] \in \R: \x_r [k] \in \mathcal{X}^{\text{dist}}_i [k] \big\}$ and $x_j^{(\ell)} [k] < x_i^{(\ell)} [k]$.
    \end{itemize}
    Ties in values are broken arbitrarily.
    Node $v_i$ then removes the state of all nodes in $\bigcup_{ \ell \in \{ 1, 2, \ldots, d \} }  \mathcal{V}_i^{\text{remove}} (\ell) [k]$ and the remaining states are stored in $\mathcal{X}_i^{\text{mm}}[k]$:
    \begin{equation}
        \mathcal{X}_i^{\text{mm}}[k] = \Big\{ \x_j [k] \in \R^d : 
        v_j \in \mathcal{V}_i^{\text{dist}} [k] \setminus \bigcup_{ \ell \in \{ 1, \ldots, d \} }  
        \mathcal{V}_i^{\text{remove}} (\ell) [k] \Big\}, 
        \label{eqn: X^mm}
    \end{equation}
    where $\mathcal{V}_i^{\text{dist}} [k] = \big\{ v_j \in \mathcal{R}: \x_j [k] \in \mathcal{X}_i^{\text{dist}} [k] \big \}$.  
    
    \item \textbf{Line 9:} $\z_i[k] \gets$ \texttt{x\_weighted\_average}($ \mathcal{X}_i^{\text{mm}} [k]$) \\
    Each node $v_i \in \mathcal{R}$ computes
    \begin{equation}
        \z_i [k] = \sum_{\x_j [k] \in \mathcal{X}_i^{\text{mm}} [k] } w_{x,ij} [k] \; \x_j [k],  \label{eqn: weight average}
    \end{equation}
    where $w_{x,ij} [k] > 0$ for all $\x_j [k] \in \mathcal{X}_i^{\text{mm}} [k]$ and $\sum_{\x_j [k] \in \mathcal{X}_i^{\text{mm}} [k]} w_{x,ij} [k] = 1$. 
    
    \item \textbf{Line 10:} $\x_i[k+1] \gets$ \texttt{gradient} ($f_i, \z_i[k]$) \\
    Node $v_i \in \mathcal{R}$ computes the gradient update as follows:
    \begin{align}
        \x_i[k+1] = \z_i[k] - \eta[k] \; \g_i[k],  \label{eqn: grad dynamic}
    \end{align}
    where $\g_i[k] \in \partial f_i( \z_i[k] )$ and 
    $\eta[k]$ is the step-size at time $k$. The conditions corresponding to the step-size are given in the next section.
    
    \item \textbf{Line 11:} $\mathcal{Y}_i^{\text{mm}}[k] \gets$ \texttt{y\_minmax\_filt}($F, \; \mathcal{Y}_i[k]$) \\
    For each dimension $\ell \in \{ 1,2, \ldots, d \}$, node $v_i \in \mathcal{R}$ removes the $F$ highest and $F$ lowest values of its neighbors' auxiliary points along that dimension.
    More specifically, for each dimension $\ell \in \{ 1, 2, \ldots, d \}$, node $v_i$ sorts the values in the set of scalars 
    $\{ y_j^{(\ell)} [k] : \y_j [k] \in \mathcal{Y}_i [k] \}$ and then removes the $F$ largest and $F$ smallest values that are larger and smaller than its own value, respectively. If there are fewer than $F$ values higher (resp. lower) than its own value, $v_i$ removes all of those values. Ties in values are broken arbitrarily. The remaining values are stored in $\mathcal{Y}_i^{\text{mm}} [k] (\ell)$ and the set $\mathcal{Y}_i^{\text{mm}} [k]$ is the collection of $\mathcal{Y}_i^{\text{mm}} [k] (\ell)$, i.e., $\mathcal{Y}_i^{\text{mm}} [k] = \big\{ \mathcal{Y}_i^{\text{mm}} [k] (\ell): \ell \in \{1, 2, \ldots, d \} \big\}$.
    
    \item \textbf{Line 12:} $\y_i[k+1] \gets$ \texttt{y\_weighted\_average}($ \mathcal{Y}_i^{\text{mm}}[k]$) \\
    For each dimension $\ell \in \{ 1, 2, \ldots, d \}$, each node $v_i \in \mathcal{R}$ computes
    \begin{equation}
        y^{(\ell)}_i [k+1] = \sum_{y^{(\ell)}_j [k] \in \mathcal{Y}_i^{\text{mm}} [k] (\ell) } w^{(\ell)}_{y,ij} [k] \; y^{(\ell)}_j[k],  \label{eqn: weight average y}
    \end{equation}
     where $w^{(\ell)}_{y,ij} [k] > 0$ for all $y^{(\ell)}_j [k] \in \mathcal{Y}_i^{\text{mm}} [k] (\ell)$ and $\sum_{y^{(\ell)}_j [k] \in \mathcal{Y}_i^{\text{mm}} [k] (\ell) } w^{(\ell)}_{y,ij} [k] = 1$. 
\end{itemize}

Note that the filtering process \texttt{x\_minmax\_filt} (\textbf{Line 8}) and the filtering process \texttt{y\_minmax\_filt} (\textbf{Line 11}) are different. In \texttt{x\_minmax\_filt}, each node removes the \textbf{whole} state vector for a neighbor if it contains an extreme value in any component, while in \texttt{y\_minmax\_filt}, each node only removes the extreme components in each vector. In addition, \texttt{x\_weighted\_average} (\textbf{Line 9}) and \texttt{y\_weighted\_average\_2} (\textbf{Line 12}) are also different in that 
\texttt{x\_weighted\_average} designates agent $v_i$ at time-step $k$ to utilize the same set of weights $\{ w_{x, ij} \in \R: \x_j[k] \in \mathcal{X}_i^{\text{mm}} [k] \}$ for all components while
\texttt{y\_weighted\_average} allows agent $v_i$ at time-step $k$ to use a different set of weights $\{ w_{y, ij}^{(\ell)} \in \R: y^{(\ell)}_j [k] \in \mathcal{Y}_i^{\text{mm}} [k] (\ell) \}$ for each coordinate $\ell$ (since the number of remaining values in each component $| \mathcal{Y}_i^{\text{mm}} [k] (\ell) |$ is not necessarily the same). These differences will become clear when considering the example provided in the next subsection.

We consider a variant of Algorithm~\ref{alg: DMM_filter} defined as follows.

\begin{algorithm}
\caption{Simultaneous Distance Filtering Dynamics} \label{alg: D_filter}
Algorithm~\ref{alg: D_filter} is the same as Algorithm~\ref{alg: DMM_filter}
except that 
\begin{itemize}
    \item \textbf{Line~8} is removed, and
    \item $\mathcal{X}_i^{\textnormal{mm}} [k]$ in \textbf{Line~9} is replaced by $\mathcal{X}_i^{\textnormal{dist}} [k]$.
\end{itemize}
\end{algorithm}

Although Algorithms~\ref{alg: DMM_filter} and \ref{alg: D_filter} are very similar (differing only in the use of an additional filter in Algorithm~\ref{alg: DMM_filter}), our subsequent analysis will reveal the relative costs and benefits of each algorithm.  We emphasize that both algorithms involve only simple operations in each iteration, and that the regular agents do not need to know the network topology, or functions possessed by another agents.
Furthermore, the regular agents do not need to know the identities of adversaries; they only need to know the upper bound for the number of local adversaries. However, we assume that all regular agents use the same step-size $\eta [k]$ (\textbf{Line 10}, equation \eqref{eqn: grad dynamic}).

\begin{remark}
While the BRIDGE framework, introduced in \cite{fang2022bridge}, encompasses several Byzantine-resilient distributed optimization algorithms, including those presented in \cite{su2016fault, sundaram2018distributed, su2020byzantine}, our proposed algorithms, namely Algorithm~\ref{alg: DMM_filter} and Algorithm~\ref{alg: D_filter}, introduce a novel concept of auxiliary states. Specifically, each regular agent $v_i$ in our algorithms maintains an auxiliary state $\y_i [k]$, updated using a consensus algorithm, placing them within the broader framework of R{\scriptsize{ED}}G{\scriptsize{RAF}} \cite{kuwaran2023geometric}.

While Algorithm~\ref{alg: DMM_filter} from our work shares a similarity with BRIDGE-T \cite{fang2022bridge} by utilizing the coordinate-wise trimmed mean, there are distinctive differences as follows. 
\begin{itemize}
    \item Firstly, our algorithm employs a distance-based filter in addition to the trimmed mean filter, allowing for an asymptotic convergence guarantee under milder assumptions (as provided in Section~\ref{subsec: assumptions}). In contrast, the convergence analysis of BRIDGE-T relies on the more restrictive assumption of i.i.d. training data.
    \item Secondly, the trimmed mean filter in BRIDGE-T eliminates both the smallest and largest $F$ values, whereas our filter discards a subset of these values, similar to the implementation in \cite{sundaram2018distributed}. This variant in our approach results in faster convergence in practice due to the resulting denser network connectivity after the filtering steps which facilitates quicker information flow \cite{sinclair1992improved}.
    \item Lastly, while BRIDGE-T uses a simple average to combine the remaining states, our algorithm employs a weighted average. These weights are chosen to satisfy Assumption~\ref{asm: weight matrices}, ensuring that the weights are lower bounded by a positive constant if the corresponding agents remain after the trimmed mean filter. This provides a more versatile and general scheme.
\end{itemize}
\end{remark}

\subsection{Example of Algorithm~\ref{alg: DMM_filter}} 

Before we prove the convergence properties of the algorithms, we first demonstrate Algorithm~\ref{alg: DMM_filter}, which is more complicated due to the min-max filtering step (\textbf{Line~8}), step by step using an example.

Suppose there are $8$ agents forming the complete graph (for the purpose of illustration). Let node $v_i$ have the local objective function $f_i: \R^2 \to \R$ defined as $f_i (\x) = ( x^{(1)} + i )^2 + ( x^{(2)} - i )^2$ for all $i \in \{1, 2, \ldots, 8 \}$. Let the set of adversarial nodes be $\mathcal{A} = \{ v_4, v_8 \}$ and thus, we have $\mathcal{R} = \{ v_1, v_2, v_3, v_5, v_6, v_7 \}$. 
Note that only the regular nodes execute the algorithm (and they do not know which agents are adversarial). Let $F = 2$ and at some time-step $\hat{k} \in \N$, each regular node has the following state and the estimated auxiliary point:\footnote{The number of agents in this demonstration is not enough to satisfy the robustness condition (Assumption~\ref{asm: robust}) presented in the next section. However, for our purpose here, it is enough to consider a small number of agents to gain an understanding for each step of the algorithm.}
\begin{align*}
    \x_1 [\hat{k}] &= \begin{bmatrix} 4 & 2 \end{bmatrix}^T, \quad
    \y_1 [\hat{k}] = \begin{bmatrix} 0 & 0 \end{bmatrix}^T, \\
    \x_2 [\hat{k}] &= \begin{bmatrix} 4 & 1 \end{bmatrix}^T, \quad
    \y_2 [\hat{k}] = \begin{bmatrix} -1 & -2 \end{bmatrix}^T, \\
    \x_3 [\hat{k}] &= \begin{bmatrix} 3 & 3 \end{bmatrix}^T, \quad
    \y_3 [\hat{k}] = \begin{bmatrix} -2 & 1 \end{bmatrix}^T, \\
    \x_5 [\hat{k}] &= \begin{bmatrix} 2 & 1 \end{bmatrix}^T, \quad
    \y_5 [\hat{k}] = \begin{bmatrix} 0 & 2 \end{bmatrix}^T, \\
    \x_6 [\hat{k}] &= \begin{bmatrix} 1 & 4 \end{bmatrix}^T, \quad
    \y_6 [\hat{k}] = \begin{bmatrix} 1 & 3 \end{bmatrix}^T, \\
    \x_7 [\hat{k}] &= \begin{bmatrix} 0 & 0 \end{bmatrix}^T, \quad
    \y_7 [\hat{k}] = \begin{bmatrix} 1 & 3 \end{bmatrix}^T.
\end{align*}
Let $\x_{a \to b} [k]$ (resp. $\y_{a \to b} [k]$) be the state (resp. estimated auxiliary point) that is sent from the adversarial node $v_a \in \mathcal{A}$ to the regular node $v_b \in \mathcal{R}$ at time-step $k$.
Suppose that in time-step $\hat{k}$, each adversarial agent sends the same states and the same estimated auxiliary points to its neighbors (although this is not necessary) as follows: 
\begin{align*}
    \x_{4 \to i} [\hat{k}] &= \begin{bmatrix} 3 & 2 \end{bmatrix}^T,
    \quad
    \y_{4 \to i} [\hat{k}] = \begin{bmatrix} -1 & 1 \end{bmatrix}^T, \\
    \x_{8 \to i} [\hat{k}] &= \begin{bmatrix} 0 & 5 \end{bmatrix}^T,
    \quad
    \y_{8 \to i} [\hat{k}] = \begin{bmatrix} 2 & 2 \end{bmatrix}^T
\end{align*}
for all $i \in \{ 1, 2, 3, 5, 6, 7 \}$. We will demonstrate the calculation of $\x_1 [\hat{k} + 1]$ and $\y_1 [\hat{k} + 1]$, computed by regular node $v_1$.

Since the network is the complete graph, the set of in-neighbors and out-neighbors of node $v_1$ is $\mathcal{N}_1^{\text{in}} = \mathcal{N}_1^{\text{out}} = \mathcal{V} \setminus \{v_1 \}$ and $\mathcal{X}_i [\hat{k}]$ (resp. $\mathcal{Y}_i [\hat{k}]$) includes all the states (resp. estimated auxiliary points). 
Then, node $v_1$ performs the distance filtering step (\textbf{Line 7}) as follows. First, it calculates the squared distances $\mathcal{D}^2_{1j} [\hat{k}]$ (since squaring does not alter the order) for all $\x_j [\hat{k}] \in \mathcal{X}_i [\hat{k}]$ as in \eqref{def: D-metric}. Node $v_1$ has 
\begin{align*}
    &\mathcal{D}^2_{11} [\hat{k}] = 20, \; 
    \mathcal{D}^2_{12} [\hat{k}] = 17, \;
    \mathcal{D}^2_{13} [\hat{k}] = 18, \;
    \mathcal{D}^2_{14} [\hat{k}] = 13, \\
    &\mathcal{D}^2_{15} [\hat{k}] = 5, \; 
    \mathcal{D}^2_{16} [\hat{k}] = 17, \;
    \mathcal{D}^2_{17} [\hat{k}] = 0, \;
    \mathcal{D}^2_{18} [\hat{k}] = 25.
\end{align*}
Since $\mathcal{D}^2_{11} [\hat{k}]$ is the second largest, node $v_1$ discards only node $v_8$'s state (which is the furthest away from $v_1$'s auxiliary point) and $\mathcal{X}_1^{\text{dist}}$ contains all states except $\x_8 [\hat{k}] = \x_{8 \to 1} [\hat{k}]$.

Then node $v_1$ performs the min-max filtering process (\textbf{Line 8}) as follows. First, consider the first component of the states in $\mathcal{X}_1^{\text{dist}}$. The states of nodes $v_1$ and $v_2$ contain the highest value in the first component (which is $4$). Since the tie can be broken arbitrarily, we choose $x^{(1)}_1 [\hat{k}]$ to come first followed by $x^{(1)}_2 [\hat{k}]$ in the ordering, so none of these values are discarded. On the other hand, the state of node $v_7$ contains the lowest value in its first component, while node $v_6$'s state contains the second lowest value in that component (since node $v_8$ has already been discarded by the distance filtering process). Node $v_1$ thus sets $\mathcal{V}^{\text{remove}}_1 (1) [\hat{k}] = \{ v_6, v_7 \}$.
Next, consider the second component in which the states of $v_6$ and $v_3$ contain the highest and second highest values, respectively, and the states of $v_7$ and $v_5$ contain the lowest and second lowest values, respectively. Thus, node $v_1$ sets $\mathcal{V}^{\text{remove}}_1 (2) [\hat{k}] = \{ v_3, v_5, v_6, v_7 \}$. Since node $v_1$ removes the entire state from all the nodes in both $\mathcal{V}^{\text{remove}}_1 (1) [\hat{k}]$ and $\mathcal{V}^{\text{remove}}_1 (2) [\hat{k}]$, according to equation \eqref{eqn: X^mm}, we have $\mathcal{X}^{\text{mm}}_1 [\hat{k}] = \big\{ \x_1 [\hat{k}], \x_2 [\hat{k}], \x_4 [\hat{k}] \big\} = \big\{ [4 \;\; 2]^T, [4 \;\; 1]^T, [3 \;\; 2]^T \big\}$.

Next, node $v_1$ performs the weighted average step (\textbf{Line 9}) as follows, Suppose node $v_1$ assigns the weights $w_{x, 11} [\hat{k}] = 0.5$, $w_{x, 12} [\hat{k}] = 0.25$ and $w_{x, 14} [\hat{k}] = 0.25$. Node $v_1$ calculates the weighted average according to \eqref{eqn: weight average} yielding $z_1^{(1)} [\hat{k}] = 3.75$ and $z_1^{(2)} [\hat{k}] = 1.75$. In the gradient step (\textbf{Line 10}), suppose $\eta [\hat{k}] = 0.1$. Node $v_1$ calculates the gradient of its local function $f_1$ at $\z_1 [\hat{k}]$ which yields $\g_1 [\hat{k}] = [9.5 \;\; 1.5]^T$ and then calculates the state $\x_1 [\hat{k} + 1]$ as described in \eqref{eqn: grad dynamic} which yields $\x_1 [\hat{k} + 1] = [2.8 \;\; 1.6]^T$.

Next, we consider the estimated auxiliary point update of node $v_1$. In fact, we can perform the update (\textbf{Line 11} and \textbf{Line 12}) for each component separately. First, consider the first component in which $v_8$ and $v_7$ contain the largest and second largest values, respectively, and $v_3$ and $v_2$ contain the smallest and second smallest values, respectively. Node $v_1$ removes these values and thus, $\mathcal{Y}^{\text{mm}}_1 [\hat{k}] (1) = \{ y^{(1)}_1 [\hat{k}], y^{(1)}_4 [\hat{k}], y^{(1)}_5 [\hat{k}], y^{(1)}_6 [\hat{k}] \} = \{ 0, -1, 0, 1 \}$. Suppose node $v_1$ assigns the weights $w^{(1)}_{y,11} [\hat{k}] = w^{(1)}_{y,14} [\hat{k}] = w^{(1)}_{y,15} [\hat{k}] = w^{(1)}_{y,16} [\hat{k}] = 0.25$. Then, the weighted average of the first component according to \eqref{eqn: weight average y} becomes $y^{(1)}_1 [\hat{k} + 1] = 0$. 
Finally, for the second component, $v_6$ and $v_7$ contain the largest values, and $v_2$ and $v_1$ contain the smallest and second smallest values, respectively. Node $v_1$ removes the value obtained from $v_2$, $v_6$ and $v_7$ and thus, the set $\mathcal{Y}^{\text{mm}}_1 [\hat{k}] (2) = \{ y^{(2)}_1 [\hat{k}], y^{(2)}_3 [\hat{k}], y^{(2)}_4 [\hat{k}], y^{(2)}_5 [\hat{k}], y^{(2)}_8 [\hat{k}] \} = \{ 0, 1, 1, 2, 2 \}$. Suppose node $v_1$ assigns the weights to each value in $\mathcal{Y}^{\text{mm}}_1 [\hat{k}] (2)$ equally. The weighted average of the second component becomes $y^{(2)}_1 [\hat{k} + 1] = 1.2$. Thus, we have $\y_1 [\hat{k} + 1] = [0 \;\; 1.2]^T$.

%% file: contents/sec-result.tex
\section{Assumptions and Main Results} 
\label{sec: assumption result}

Having defined the steps in Algorithms~\ref{alg: DMM_filter} and \ref{alg: D_filter}, we now turn to proving their resilience and convergence properties.

\subsection{Assumptions}
\label{subsec: assumptions}

\begin{assumption} \label{asm: convex}
For all $v_i \in \mathcal{V}$, the functions $f_i(\x)$ are convex, and the sets $\argmin f_i(\x)$ are non-empty and bounded.  
\end{assumption}

Since the set $\argmin f_i(\x)$ is non-empty, let $\x_i^*$ be an arbitrary minimizer of the function $f_i$. 

\begin{assumption} \label{asm: gradient_bound}
There exists $L \in \R_{>0}$ such that 
$\| \tilde{\g}_i( \x ) \|_2 \leq L$
for all $\x \in \R^d$, $v_i \in \mathcal V$, and $\tilde{\g}_i ( \x ) \in \partial f_i(\boldsymbol{ \x })$.
\end{assumption}

The bounded subgradient assumption above is common in the distributed convex optimization literature \cite{nedic2009distributed, duchi2011dual, jakovetic2014fast}. 

\begin{assumption} \label{asm: step-size}
The step-size sequence $\{ \eta[k] \}_{k=0}^{\infty} \subset \R_{> 0}$ used in \textbf{Line 11} of Algorithm~\ref{alg: DMM_filter} is of the form 
\begin{equation}
    \eta [k] = \frac{c_1}{k + c_2} \quad \text{for some} \;\; c_1, c_2 \in \R_{>0}. \label{eqn: step-size}
\end{equation}
\end{assumption}

Note that the step-size in \eqref{eqn: step-size} satisfies $\eta[k+1] < \eta[k]$ for all $k \in \N$, and
\begin{equation}
    \lim_{k \to \infty} \eta[k] = 0 \quad \text{and} \quad
    \sum_{k=0}^{\infty} \eta [k] = \infty  \label{eqn: step-size limit}
\end{equation}
for any choices of $c_1$, $c_2 \in \R_{>0}$.

\begin{assumption} \label{asm: robust}
Given a positive integer $F \in \Z_+$, the Byzantine agents form a $F$-local set.
\end{assumption}

\begin{assumption} \label{asm: weight matrices}
For all $k \in \N$ and $\ell \in \{ 1, 2, \ldots, d \}$, the weights $w_{x,ij} [k]$ and $w^{(\ell)}_{y,ij} [k]$ (used in \textbf{Line 9} and \textbf{Line 12} of Algorithm \ref{alg: DMM_filter}) are positive if and only if $\x_j [k] \in \mathcal{X}_i^{\textup{mm}} [k]$ for Algorithm~\ref{alg: DMM_filter} (and $\x_j [k] \in \mathcal{X}_i^{\textup{dist}} [k]$ for Algorithm~\ref{alg: D_filter}) and $y^{(\ell)}_j [k] \in \mathcal{Y}_i^{\textup{mm}} [k] (\ell)$, respectively. Furthermore, there exists $\omega \in \R_{>0}$ such that for all $k \in \N$ and $\ell \in \{ 1, 2, \ldots, d \}$, the non-zero weights are lower bounded by $\omega$.
\end{assumption}

\begin{remark}  \label{rem: F knowledge}
Regarding the prior knowledge of $F$ in Assumption~\ref{asm: robust}, we note that, as with any reliable or secure system, one has to design the system to provide a desired degree of reliability. If one requires the system to provide resilience to a certain number of faulty nodes, one has to design the algorithm (and network) to facilitate that. This is the standard philosophy and methodology in the literature \cite{lynch1996distributed, blanchard2017machine, yang2019byrdie}. Note that $F$ does not have to be the exact number of adversarial nodes -- it is only an upper bound on the number of adversarial nodes locally.
\end{remark}

\subsection{Analysis of Auxiliary Point Update}

Since the dynamics of the estimated auxiliary points $\{ \y_i[k] \}_{\mathcal{R}}$ are independent of the dynamics of the estimated solutions $\{ \x_i[k] \}_{\mathcal{R}}$, we begin by analyzing the convergence properties of the estimated auxiliary points $\{ \y_i[k] \}_{\mathcal{R}}$.

In order to establish this result, we need to define the following scalar quantities. For $k \in \N$ and $\ell \in \{1,2, \ldots, d\}$, let $M^{(\ell)}[k] := \max_{v_i \in \mathcal{R}} y_i^{(\ell)}[k]$, $m^{(\ell)}[k] := \min_{v_i \in \mathcal{R}} y_i^{(\ell)}[k]$, and $D^{(\ell)}[k] := M^{(\ell)}[k] - m^{(\ell)}[k]$.
Define the vector $\boldsymbol{D}[k] := \big[ D^{(1)}[k], D^{(2)}[k], \cdots, D^{(d)}[k] \big]^T$.

The proposition below shows that the estimated auxiliary points $\{ \y_i[k] \}_{\mathcal{R}}$ converge \textbf{exponentially fast} to a single point called $\y[\infty]$. 
 
\begin{proposition} \label{prop: aux convergence}
Suppose Assumption~4 hold, the graph $\mathcal{G}$ is $(2F + 1)$-robust, and the weights $w^{(\ell)}_{y,ij} [k]$ satisfy Assumption~\ref{asm: weight matrices}.
Suppose the estimated auxiliary points of the regular agents $\{ \y_i[k] \}_{\mathcal{R}}$ follow the update rule described as \textbf{Line~11} and \textbf{Line~12} in Algorithm~\ref{alg: DMM_filter}. Then, in both Algorithm~\ref{alg: DMM_filter} and Algorithm~\ref{alg: D_filter}, 
there exists $\y[\infty] \in \R^d$ with $y^{(\ell)}[\infty] \in \big[ m^{(\ell)}[k], M^{(\ell)}[k] ]$ for all $k \in \N$ and $\ell \in \{1, 2, \ldots, d \}$ such that for all $v_i \in \mathcal{R}$, we have
\begin{align*}
    \| \y_i[k] - \y[\infty] \| < \beta e^{-\alpha k},
\end{align*}
where $\alpha := \frac{1}{| \mathcal{R} | - 1} \log \frac{1}{\gamma} > 0$, $\beta := \frac{1}{\gamma} \| \boldsymbol{D} [0] \|$, and $\gamma :=  1 - \frac{\omega^{ | \mathcal{R} | - 1}}{2}$.
\end{proposition}

The proof of the above proposition follows by noting that the updates for $\{ \y_i[k] \}_{\mathcal{R}}$ essentially boil down to a set of $d$ scalar consensus updates (one for each dimension of the vector), Thus, one can directly leverage the proof for scalar consensus (with filtering of extreme values) from \cite[Proposition~6.3]{sundaram2018distributed}. We provide the proof of Proposition~\ref{prop: aux convergence} in Appendix~\ref{sec: proof of aux prop}. 

Recall that $\{ \hat{\x}_i^* \}_\mathcal{R}$ is the set containing the approximate minimizers of the regular nodes' local functions.
Let $\x$ be a matrix in $\R^{d \times |\mathcal{R}|}$, where each column of $\x$ is a different vector from $\{ \hat{\x}_i^* \}_\mathcal{R}$.
In addition, let $\overline{\x}$ and $\underline{\x}$ be the vectors in $\R^d$ defined by $\overline{x}_i = \max_{1 \leq j \leq |\mathcal{R}|} [\x]_{ij}$ and $\underline{x}_i = \min_{1 \leq j \leq |\mathcal{R}|} [\x]_{ij}$, respectively.
 Since we set $\y_i[0] = \hat{\x}_i^*$ for all $v_i \in \mathcal{R}$ according to \textbf{Line 2} in Algorithm~\ref{alg: DMM_filter}, we can write 
\begin{equation*}
    \beta = \frac{1}{\gamma} \| \boldsymbol{D} [0] \|
    = \frac{1}{\gamma} \| \overline{\x} - \underline{\x} \|.
\end{equation*}

\subsection{Convergence to Consensus of States}

Having established convergence of the auxiliary points to a common value (for the regular nodes), we now consider the state update and show that the states of all regular nodes $\{ \x_i [k] \}_{\mathcal{R}}$ asymptotically reach consensus under Algorithm \ref{alg: DMM_filter}. Before stating the main theorem, we provide a result from \cite[Lemma~2.3]{sundaram2018distributed} which is important for proving the main theorem.

\begin{lemma}
Suppose the graph $\mathcal{G}$ satisfies Assumption~\ref{asm: robust} and is $((2d + 1)F + 1)$-robust. Let $\mathcal{G}'$ be a graph obtained by removing $(2d + 1)F$ or fewer incoming edges from each node in $\mathcal{G}$. Then $\mathcal{G}'$ is rooted.
\label{lem: rooted subgraph}
\end{lemma}

This means that if we have enough redundancy in the network (in this case, captured by the $( (2d + 1)F + 1 )$-robustness condition), information from at least one node can still flow to the other nodes in the network even after each regular node discards up to $F$ neighboring states in the distance filtering step  (\textbf{Line~7}) and up to $2dF$ neighboring states in the min-max filtering step  (\textbf{Line~8}). This transmissibility of information is a crucial condition for reaching consensus among regular nodes.

\begin{theorem}[Consensus] \label{thm: consensus}
Suppose Assumptions~\ref{asm: gradient_bound}-\ref{asm: weight matrices} hold, and the graph $\mathcal{G}$ is $((2d + 1)F + 1)$-robust. If the regular agents follow Algorithm~\ref{alg: DMM_filter} then for all $v_i, v_j \in \mathcal{R}$, it holds that
\begin{equation*}
    \lim_{k \to \infty} \| \x_{i} [k] - \x_{j} [k] \| = 0.
\end{equation*}
\end{theorem}

\begin{proof}
It is sufficient to show that all regular nodes $v_i \in \mathcal{R}$ reach consensus on each component of their vectors $\x_i[k]$ as $k \to \infty$.  For all $\ell \in \{ 1, 2, \ldots, d \}$ and for all $v_i \in \mathcal{R}$, from \eqref{eqn: weight average} and \eqref{eqn: grad dynamic}, the $\ell$-th component of the vector $\x_i [k]$ evolves as
\begin{equation*}
    x^{(\ell)}_i [k+1] = \sum_{\x_j [k] \in \mathcal{X}_i^{\text{mm}} [k]} w_{x,ij} [k] \; x^{(\ell)}_j [k]
    - \eta[k] \; g^{(\ell)}_i [k].
\end{equation*}
From \cite[Proposition~5.1]{sundaram2018distributed}, the above equation can be rewritten as
\begin{equation}
    x^{(\ell)}_i [k+1] = \sum_{v_j \in (\mathcal{N}^{\text{in}}_i \cap \mathcal{R}) \cup \{v_i\}} \bar{w}^{(\ell)}_{x,ij} [k] \; x^{(\ell)}_j [k]
    - \eta[k] \; g^{(\ell)}_i [k],
    \label{eqn: regular dynamics}
\end{equation}
where $\bar{w}^{(\ell)}_{x,ii} [k] + \sum_{v_j \in \mathcal{N}^{\text{in}}_i \cap \mathcal{R}} \bar{w}^{(\ell)}_{x,ij} [k] = 1$, and 
$\bar{w}^{(\ell)}_{x,ii} [k] > \omega$ and at least $| \mathcal{N}^{\text{in}}_i | - 2F$ of the other weights are lower bounded by $\frac{\omega}{2}$.

Consider the set $\mathcal{X}_i^{\text{mm}} [k]$ which is obtained by removing at most $F + 2dF$ states received from $v_i$'s neighbors (up to $F$ states removed by the distance filtering process in \textbf{line 7}, and up to $2F$ additional states removed by the min-max filtering process on each of the $d$ components in \textbf{line 8}).  
Since the graph is $((2d+1)F + 1)$-robust and the Byzantine agents form an $F$-local set by Assumption~\ref{asm: robust}, from Lemma~\ref{lem: rooted subgraph}, the subgraph consisting of regular nodes will be rooted.
Using the fact that the term $\eta[k] \; g^{(\ell)}_i [k]$ asymptotically goes to zero (by Assumptions~\ref{asm: gradient_bound} and \eqref{eqn: step-size limit}) and equation \eqref{eqn: regular dynamics}, we can proceed as in the proof of \cite[Theorem~6.1]{sundaram2018distributed} to show that
\begin{equation*}
    \lim_{k \to \infty} | x^{(\ell)}_{i} [k] - x^{(\ell)}_{j} [k] | = 0,
\end{equation*}
for all $v_i, v_j \in \mathcal{R}$, which completes the proof.
\end{proof}

Theorem~\ref{thm: consensus} established consensus of the states of the regular agents, leveraging (and extending) similar analysis for scalar functions from \cite{sundaram2018distributed}, only for Algorithm~\ref{alg: DMM_filter}. 
However, this does not hold for Algorithm~\ref{alg: D_filter} since there might exist a regular agent $v_i \in \mathcal{R}$, time-step $k \in \N$ and dimension $\ell \in \{ 1, 2, \ldots, d \}$ such that
an adversarial state $x_s^{(\ell)}[k] \in \{ x_j^{(\ell)}[k] \in \R: \x_j[k] \in \mathcal{X}_i^{\text{dist}}[k], v_j \in \mathcal{A} \}$ 
cannot be written as a convex combination of $\big \{ x_j^{(\ell)}[k] \in \R: v_j \in (\mathcal{N}^{\text{in}}_i \cap \mathcal{R}) \cup \{v_i\} \big\}$, and thus we cannot obtain equation \eqref{eqn: regular dynamics}.
On the other hand, Proposition~\ref{prop: aux convergence} established consensus of the auxiliary points, which will be now used to characterize the convergence {\it region} of both Algorithm~\ref{alg: DMM_filter} and Algorithm~\ref{alg: D_filter}.

\subsection{The Region To Which The States Converge}

We now analyze the trajectories of the states of the agents under Algorithm~\ref{alg: DMM_filter} and Algorithm~\ref{alg: D_filter}. 
We start with the following result regarding the intermediate state $\z_i[k]$ calculated in \textbf{Lines 7-9} of Algorithm~\ref{alg: DMM_filter}. 

\begin{lemma} \label{lem: dist to curr aux}
Suppose Assumptions~\ref{asm: robust} and \ref{asm: weight matrices} hold. Furthermore:
\begin{itemize}
    \item if the regular agents follow Algorithm~\ref{alg: DMM_filter}, suppose the graph $\mathcal{G}$ is $((2d + 1)F + 1)$-robust;
    \item otherwise, if the regular agents follow Algorithm~\ref{alg: D_filter}, suppose the graph $\mathcal{G}$ is $(2F + 1)$-robust.
\end{itemize}
\noindent For all $k \in \N$ and $v_i \in \mathcal{R}$, if there exists $R_i[k] \in \R_{\ge 0}$ such that $ \| \x_j[k] - \y_i[k] \| \leq R_i[k]$ for all $v_j \in (\mathcal{N}_i^{\textup{in}} \cap \mathcal{R}) \cup \{ v_i \}$ then $ \| \z_i[k] - \y_i[k] \| \leq R_i[k]$.
\end{lemma}  

\begin{proof}
Consider the distance filtering step in \textbf{Line 7} of Algorithm \ref{alg: DMM_filter}. Recall the definition of $\mathcal{D}_{ij} [k]$ from \eqref{def: D-metric}.
We will first prove the following claim. For each $k \in \N$ and $v_i \in \mathcal{R}$, there exists $v_r \in ( \mathcal{N}_i^{\text{in}} \cap \mathcal{R} ) \cup \{ v_i \}$ such that for all $\x_j [k] \in \mathcal{X}^{\text{dist}}_i [k]$,
\begin{equation*}
    \| \x_j [k] - \y_i [k] \| \leq \| \x_r [k] - \y_i [k] \|,
\end{equation*}
or equivalently, $\mathcal{D}_{ij} [k] \leq \mathcal{D}_{ir} [k]$.

There are two possible cases.  First, if the set $\mathcal{X}_i^{\text{dist}} [k]$ contains only regular nodes, we can simply choose $v_r \in ( \mathcal{N}_i^{\text{in}} \cap \mathcal{R} ) \cup \{ v_i \}$ to be the node whose state $\x_r [k]$ is furthest away from $\y_i [k]$.  
Next, consider the case where $\mathcal{X}_i^{\text{dist}} [k]$ contains the state of one or more Byzantine nodes. Since node $v_i \in \mathcal{R}$ removes the $F$ states from $\mathcal{N}_{i}^{\text{in}}$ that are furthest away from $\y_i [k]$ (\textbf{Line 7}), 
and there are at most $F$ Byzantine nodes in $\mathcal{N}_i^{\text{in}}$, there is at least one regular state removed by node $v_i$. Let $v_r$ be one of the regular nodes whose state is removed. We then  have $\mathcal{D}_{ir} [k] \geq \mathcal{D}_{ij} [k]$, for all $v_j \in \{ v_s \in \mathcal{V}: \x_s [k] \in \mathcal{X}_i^{\text{dist}} [k] \}$ which proves the claim.  

If Algorithm~\ref{alg: DMM_filter} is implemented, let $\hat{\mathcal{X}}_i [k] = \mathcal{X}_i^{\text{mm}}[k]$ and we have that $\mathcal{X}_i^{\text{mm}}[k] \subseteq \mathcal{X}_i^{\text{dist}}[k]$ due to the min-max filtering step in \textbf{Line~8}. If Algorithm~\ref{alg: D_filter} is implemented, let $\hat{\mathcal{X}}_i [k] = \mathcal{X}_i^{\text{dist}}[k]$ since \textbf{Line~8} is removed.
Then, consider the weighted average step in \textbf{Line 9}. From \eqref{eqn: weight average}, we have
\begin{equation*}
    \z_i [k] - \y_i [k]
    = \sum_{ \x_j [k] \in \hat{\mathcal{X}}_i [k] }  w_{x, ij} [k]  \; \big( \x_j [k] - \y_i [k] \big) .
\end{equation*}
Since $\| \x_j [k] - \y_i [k] \| \leq \| \x_r [k] - \y_i [k] \|$ for all $\x_j [k] \in \hat{\mathcal{X}}_i [k]$ (where $v_r$ is the node identified in the claim at the start of the proof), we obtain
\begin{equation*}
    \| \z_i [k] - \y_i [k] \|
    \leq \sum_{ \x_j [k] \in \hat{\mathcal{X}}_i [k] }
    w_{x, ij} [k] \;  \| \x_j [k] - \y_i [k] \| 
    \leq \sum_{ \x_j [k] \in \hat{\mathcal{X}}_i [k] }
    w_{x, ij} [k] \; \| \x_r [k] - \y_i [k] \|. 
\end{equation*}
Since $v_r \in ( \mathcal{N}_i^{\text{in}} \cap \mathcal{R} ) \cup \{ v_i \}$, by our assumption, we have $\| \x_r [k] - \y_i [k] \| \leq R_i[k]$. Thus, using the above inequality and Assumption \ref{asm: weight matrices}, we obtain that $ \| \z_i [k] - \y_i [k] \| \leq R_i[k]$.
\end{proof}

Lemma~\ref{lem: dist to curr aux} essentially states that if the set of states $\{ \x_j[k]: v_j \in (\mathcal{N}_i^{\text{in}} \cap \mathcal{R} ) \cup \{v_i\}  \}$ is a subset of the local ball $\B ( \y_i [k] , R_i [k] )$ then the intermediate state $\z_i [k]$ is still in the ball. This is a consequence of using the distance filter (and adding the min-max filter in Algorithm~\ref{alg: DMM_filter} does not destroy this property), and this will play an important role in proving the convergence theorem.

Next, we will establish certain quantities that will be useful for our analysis of the convergence region. For $v_i \in \mathcal{R}$ and $\epsilon > 0$, define
\begin{align}
    \mathcal{C}_i (\epsilon) := \{ \x \in \R^d : f_i(\x) \leq f_i(\x_i^*) + \epsilon \}.   \label{def: sublevel_set}
\end{align}
For all $v_i \in \mathcal{R}$, since the set $\argmin f_i(\x)$ is bounded  (by Assumption~\ref{asm: convex}), there exists $\delta_i (\epsilon) \in (0, \infty)$ such that
\begin{equation}
    \mathcal{C}_i( \epsilon ) \subseteq \B(\x_i^*, \delta_i (\epsilon)).
    \label{def: delta_i}
\end{equation}

The following proposition, whose proof is provided in Appendix~\ref{sec: proof_grad_angle}, introduces an angle $\theta_i$ which is an upper bound on the angle between the negative of the gradient of $f_i$ at a given point $\x$ and the vector $\x_i^* - \x$.
\begin{proposition} \label{prop: grad angle}
If Assumptions~\ref{asm: convex} and \ref{asm: gradient_bound} hold
then for all $v_i \in \mathcal{R}$ and $\epsilon > 0$, there exists $\theta_i (\epsilon) \in \big[ 0, \frac{\pi}{2} \big)$ such that for all $\x \notin \mathcal{C}_i( \epsilon )$ and $\tilde{\g}_i ( \x ) \in \partial f_i(\boldsymbol{ \x })$,
\begin{equation}
    \angle ( - \tilde{\g}_i (\x), \; \x_i^* - \x ) \leq \theta_i (\epsilon).
    \label{def: theta_i}
\end{equation}
\end{proposition}


Before stating the main theorem, we define 
\begin{equation}
    \tilde{R}_i := \| \x_i^* - \y [\infty] \|
    \label{def: tilde R_i}.
\end{equation}
Furthermore, for all $\xi \in \R_{\geq 0}$ and $\epsilon \in \R_{> 0}$, we define the \textit{convergence radius}
\begin{equation}
    s^*( \xi, \epsilon ) := \max_{v_i \in \mathcal{R}} \big\{ \max \{ \Tilde{R}_i \sec \theta_i(\epsilon), \Tilde{R}_i + \delta_i(\epsilon) \} \big\} + \xi.   
    \label{def: conv radius}
\end{equation}
where $\tilde{R}_i$, $\theta_i (\epsilon)$ and $\delta_i (\epsilon)$ are defined in \eqref{def: tilde R_i}, \eqref{def: theta_i} and \eqref{def: delta_i}, respectively. Based on the definition above, we refer to $\B (\y [\infty], s^*( \xi, \epsilon) )$ as the \textit{convergence ball}.

We now come to the main result of this paper, showing that the states of all the regular nodes will converge to a ball of radius $\inf_{\epsilon > 0} s^*(0, \epsilon)$ around the auxiliary point $\y[\infty]$ under Algorithm~\ref{alg: DMM_filter} and Algorithm~\ref{alg: D_filter}.

\begin{theorem}[Convergence] \label{thm: convergence}
Suppose Assumptions~\ref{asm: convex}-\ref{asm: weight matrices} hold. Furthermore:
\begin{itemize}
    \item if the regular agents follow Algorithm~\ref{alg: DMM_filter}, suppose the graph $\mathcal{G}$ is $((2d + 1)F + 1)$-robust;
    \item otherwise, if the regular agents follow Algorithm~\ref{alg: D_filter}, suppose the graph $\mathcal{G}$ is $(2F + 1)$-robust.
\end{itemize}
\noindent Then regardless of the actions of any $F$-local set of Byzantine adversaries, for all $v_i \in \mathcal{R}$, we have
\begin{align*}
    \limsup_k \| \x_i[k] - \y[\infty] \| \leq \inf_{\epsilon > 0} s^*(0, \epsilon).
\end{align*}
\end{theorem}

The proof of the theorem requires several technical lemmas and propositions, and thus, we provide a proof sketch in Section~\ref{subsec: proof sketch of conv} and a formal proof in the appendix. 

The following theorem, whose proof is provided in Appendix~\ref{sec: proof_true_sol}, provides possible locations of the true minimizer $\x^*$, which is in fact inside the convergence region, even in the presence of adversarial agents.
\begin{theorem} \label{thm: true_sol}
Let $\x^*$ be a solution of Problem~\eqref{prob: regular node}. If Assumptions~\ref{asm: convex} and \ref{asm: gradient_bound} hold, then $\x^* \in \B \big( \y [\infty], \; \inf_{\epsilon > 0} s^*( 0, \epsilon) \big)$.
\end{theorem}

Theorem~\ref{thm: convergence} and Theorem~\ref{thm: true_sol} show that both Algorithms~\ref{alg: DMM_filter} and \ref{alg: D_filter} cause all regular nodes to converge to a region that also contains the true solution, regardless of the actions of any $F$-local set of Byzantine adversaries.  The size of this region scales with the quantity $\inf_{\epsilon > 0} s^*(0, \epsilon)$. Loosely speaking, this quantity becomes smaller as the minimizers of the local functions of the regular agents get closer together.  
More specifically, consider a fixed $\epsilon \in \R_{>0}$. If the functions $f_i (\x)$ are translated so that the minimizers $\x_i^*$ get closer together (i.e., $\Tilde{R}_i$ is smaller while $\theta_i (\epsilon)$ and $\delta_i (\epsilon)$ are fixed), then $s^*(0, \epsilon)$ also decreases. Consequently, the state $\x_i[k]$ is guaranteed to become closer to the true minimizer $\x^*$ as $k$ goes to infinity. Figure~\ref{fig: quantity illus} illustrates the key quantities outlined in the main theorems. A detailed discussion of the convergence region is further provided in Section \ref{subsec: convergence ball}.

\begin{figure}[ht]
    \centering
    \includegraphics[width=0.45\textwidth]{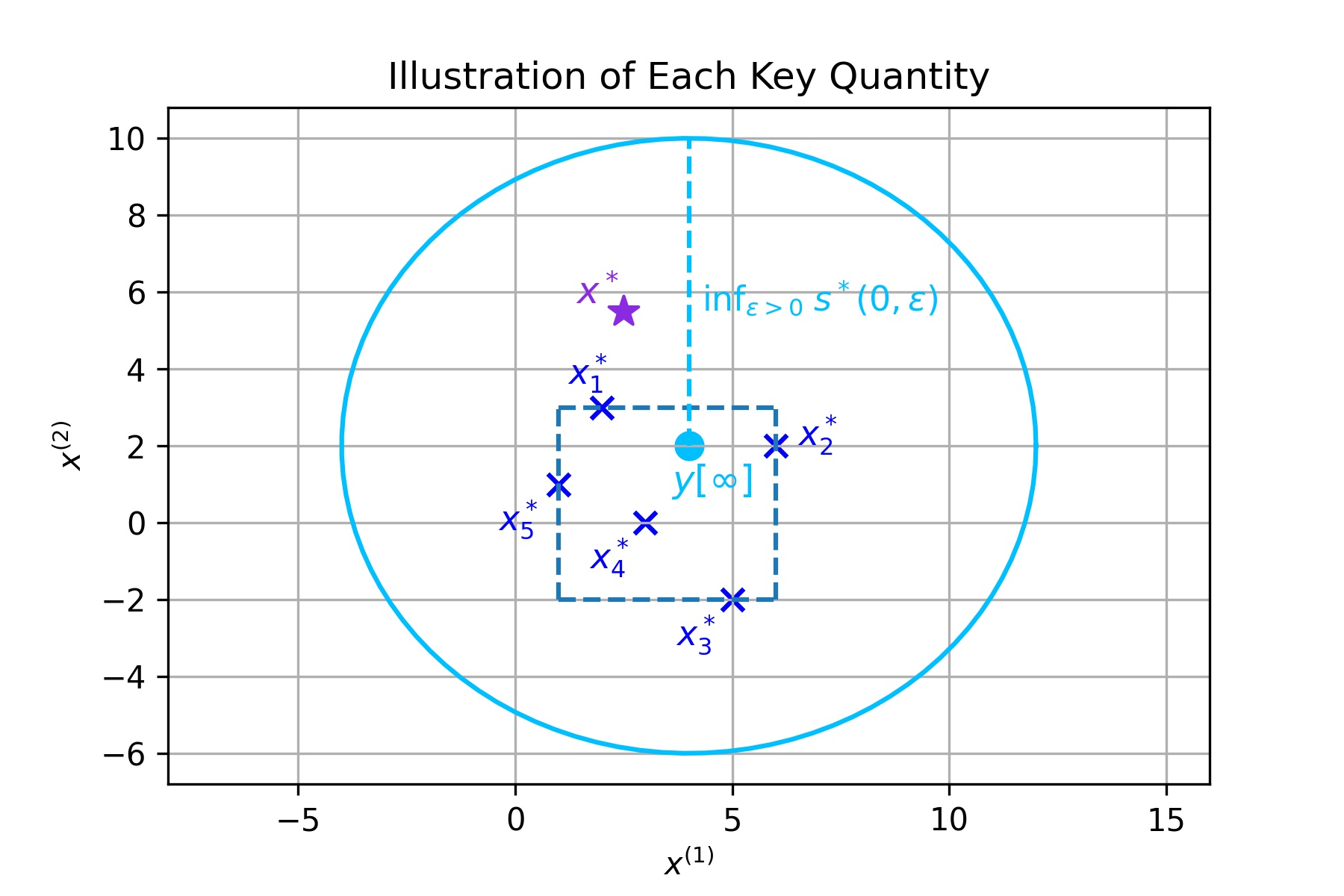}
    \caption{The local minimizers $\x_i^*$ and the global minimizer $\x^*$ are shown in the plot. The estimated auxiliary point $\y [\infty]$ is in the rectangle formed by the local minimizers (Proposition~\ref{prop: aux convergence}) whereas the global minimizer $\x^*$ is not necessarily in the rectangle \cite{kuwaran2018location}. However, the ball centered at $\y [\infty]$ with radius $\inf_{\epsilon > 0} s^*(0, \epsilon)$ contains both the supremum limit of the state vectors $\x_i[k]$ and the global minimizer $\x^*$ (Theorem~\ref{thm: convergence} and \ref{thm: true_sol}).}
    \label{fig: quantity illus}
\end{figure}

We would like to highlight the scalability of our algorithms in terms of both computational complexity and graph robustness requirements, specifically in relation to the number of dimensions $d$. Algorithms~\ref{alg: DMM_filter} and \ref{alg: D_filter} exhibit computational complexities of $\Tilde{\mathcal{O}}(d^2)$ and $\Tilde{\mathcal{O}}(d)$ operations per agent per iteration, respectively. A detailed calculation is provided in Section~\ref{subsec: complexity}. Furthermore, they impose robustness requirements of $\mathcal{O}(d)$ and $\mathcal{O}(1)$ to achieve the convergence result, as demonstrated in Theorem~\ref{thm: convergence}. While Algorithm~\ref{alg: D_filter} is more scalable, it lacks a consensus guarantee, unlike Algorithm~\ref{alg: DMM_filter} (refer to Theorem~\ref{thm: consensus}). Further insights and discussions on this topic are presented in Section~\ref{subsec: tradeoff} and Remark~\ref{rem: redundancy}.

\subsection{Proof Sketch of the Convergence Theorem} 
\label{subsec: proof sketch of conv}

We work towards the proof of Theorem \ref{thm: convergence} in several steps,  which we provide an overview below. The proofs of the intermediate results presented in this section are provided in Appendices~\ref{sec: proof_main_1}, \ref{sec: proof_main_2} and \ref{sec: proof_main_3}.  

For the subsequent analysis, we suppose that the graph $\mathcal{G}$
\begin{itemize}
    \item is $((2d + 1)F + 1)$-robust for Algorithm~\ref{alg: DMM_filter}, and
    \item is $(2F + 1)$-robust for Algorithm~\ref{alg: D_filter}.
\end{itemize}
Furthermore, unless stated otherwise, we will fix $\xi \in \R_{>0}$ and $\epsilon \in \R_{>0}$, and hide the dependence of $\xi$ and $\epsilon$ in $\delta_i (\epsilon)$ and $s^*( \xi, \epsilon )$ by denoting them as $\delta_i$ and $s^*$, respectively.

\subsubsection{Gradient Update Step Analysis}  
\label{subsec: grad_update analysis}

First, we consider the update from the intermediate states $\{ \z_i[k] \}_{\mathcal{R}}$ to the states $\{ \x_i[k+1] \}_{\mathcal{R}}$ via the gradient step \eqref{eqn: grad dynamic} (i.e., \textbf{Line~10}).
In particular, we provide a relationship between $\| \z_i [k] - \y [\infty] \|$ and $\| \x_i [k+1] - \y [\infty] \|$ for three different cases:
\begin{itemize}
    \item $\| \z_i [k] - \y [\infty] \| \in \big[ 0, \; \max_{v_j \in \mathcal{R}} \{ \tilde{R}_j + \delta_j \} \big]$,
    \item $\| \z_i [k] - \y [\infty] \| \in \big( \max_{v_j \in \mathcal{R}} \{ \tilde{R}_j + \delta_j \}, \; s^* \big]$,
    \item $\| \z_i [k] - \y [\infty] \| \in ( s^* , \infty )$.
\end{itemize}

The corresponding formal statements are presented as follows.
Lemma~\ref{lem: dist grad update 1} below essentially says that if $k$ is sufficiently large and $\z_i[k] \in \B( \y[\infty], \max_{v_i \in \mathcal{R}} \{ \Tilde{R}_i + \delta_i \} )$, then after applying the gradient update \eqref{eqn: grad dynamic}, the state $\x_i[k + 1]$ will still be in the convergence ball.
To establish the result, let $k^*_1 \in \N$ be a time-step such that $\eta[k^*_1] \leq \frac{\xi}{L}$.
\begin{lemma}  \label{lem: dist grad update 1} 
Suppose Assumptions~\ref{asm: gradient_bound}-\ref{asm: weight matrices} hold.
For all $v_i \in \mathcal{R}$ and $k \geq k^*_1$, if 
$\z_i[k] \in \B \big( \y[\infty], \; \max_{v_j \in \mathcal{R}} \{ \Tilde{R}_j + \delta_j \} \big)$ then $\x_i[k+1] \in \B (\y [\infty], s^*)$.  
\end{lemma}

Lemma~\ref{lem: grad update}, based on Proposition~\ref{prop: grad angle}, analyzes the relationship between $\| \z_i[k] - \y[\infty] \|$ and $\| \x_i[k+1] - \y[\infty] \|$ when $\| \z_i[k] - \y[\infty] \| > \Tilde{R}_i + \delta_i$. The result will be used to prove Lemma \ref{lem: dist grad update 2}. 

For $v_i \in \mathcal{R}$, define $\Delta_i: [ \Tilde{R}_i, \infty) \times \R_{\geq 0} \to \R$ to be the function
\begin{equation}
    \Delta_i( p, l ) := 2 l  \big(\sqrt{ p^2 - \Tilde{R}_i^2 } \cos \theta_i - \Tilde{R}_i \sin \theta_i \big) - l^2.  
    \label{def: delta}
\end{equation}
\begin{lemma}  \label{lem: grad update}
Suppose Assumptions~\ref{asm: convex}, \ref{asm: gradient_bound}, \ref{asm: robust} and \ref{asm: weight matrices} hold.
For all $v_i \in \mathcal{R}$ and $k \in \N$, if $\| \z_i[k] - \y[\infty] \| > \Tilde{R}_i + \delta_i$ then 
\begin{equation}
    \| \x_i[k+1] - \y[\infty] \|^2 
    \leq \| \z_i[k] - \y[\infty] \|^2 
    - \Delta_i( \| \z_i[k] - \y[\infty] \|, \; \eta[k] \; \| \g_i[k] \| ),
    \label{eqn: bound second case}
\end{equation}  
where $\g_i[k] \in \R^d$ is defined in \eqref{eqn: grad dynamic}.
\end{lemma}

Similar to Lemma~\ref{lem: dist grad update 1}, Lemma~\ref{lem: dist grad update 2} below states that if $k$ is sufficiently large and $\| \z_i[k] - \y[\infty] \| \in \big( \max_{v_i \in \mathcal{R}} \{ \Tilde{R}_i + \delta_i \}, \; s^* \big]$ then by applying the gradient step \eqref{eqn: grad dynamic}, we have that the state $\x_i[k + 1]$ is still in the convergence ball.

To simplify the notations, define
\begin{equation}
    a_i^\pm := -\Tilde{R}_i \sin \theta_i \pm \sqrt{ ( s^* )^2 - \Tilde{R}_i^2 \cos^2 \theta_i }  
    \quad \text{and} \quad
    b_i := 2 \big( \sqrt{ ( s^* )^2 - \Tilde{R}_i^2 } \cos \theta_i - \Tilde{R}_i \sin \theta_i \big).  
    \label{def: a_i & b_i}
\end{equation}
Let $k_2^* \in \N$ be a time-step such that $\eta[k_2^*] \leq \frac{1}{L} \min_{v_i \in \mathcal{R}} \big\{ \min \{ a_i^+, \; b_i \} \big\}$.
\begin{lemma}  \label{lem: dist grad update 2}
Suppose Assumptions~\ref{asm: convex}-\ref{asm: weight matrices} hold.
For all $v_i \in \mathcal{R}$ and $k \geq k_2^*$, 
if $\| \z_i[k] - \y[\infty] \| \in \big( \max_{v_j \in \mathcal{R}} \{ \Tilde{R}_j + \delta_j \}, \; s^* \big]$ 
then $\| \x_i[k+1] - \y [\infty] \| \in [0, \; s^*]$.   
\end{lemma}

The following lemma is useful for bounding the term $\Delta_i$ appeared in \eqref{eqn: bound second case} for the case that $\| \z_i[k] - \y[\infty] \| > s^*$.

Define the set of agents
\begin{equation}
    \mathcal{I}_z[k] := \{ v_i \in \mathcal{R} : \| \z_i[k] - \y[\infty] \| > s^* \}, 
    \label{def: I outside s}
\end{equation}
and let $k_3^* \in \N$ be a time-step such that $\eta[k_3^*] \leq \frac{1}{2L} \min_{v_i \in \mathcal{R}} b_i$.
\begin{lemma}  \label{lem: Delta lower bound}
If Assumptions~\ref{asm: convex}-\ref{asm: weight matrices} hold then
for all $k \geq k^*_3$ and $v_i \in \mathcal{I}_z[k]$, 
\begin{align*}
    \Delta_i( \| \z_i[k] - \y[\infty] \|, \; \eta[k] \; \| \g_i[k] \| ) > \frac{1}{2} b_i \underline{L}_i \eta[k],
\end{align*}
where $\Delta_i$ and $\g_i[k]$ are defined in \eqref{def: delta} and \eqref{eqn: grad dynamic}, respectively, and $\underline{L}_i := \frac{\epsilon}{\delta_i(\epsilon)} > 0$.
\end{lemma}

Note that the quantity $\underline{L}_i$ defined above can be interpreted as a lower bound on a subgradient of the function $f_i (\x)$ when $\x \notin \mathcal{C}_i (\epsilon)$. 

Lemmas~\ref{lem: dist grad update 1}-\ref{lem: Delta lower bound} collectively establish the complete relationship governing the update from $\{ \z_i[k] \}_{\mathcal{R}}$ to $\{ \x_i[k+1] \}_{\mathcal{R}}$, which will be used to prove Lemma~\ref{lem: two bounds}.

\subsubsection{Bounds on States of Regular Agents}  
 \label{subsec: bounds}

Next, we consider the update from the states $\{ \x_i[k] \}_{\mathcal{R}}$ to the intermediate states $\{ \z_i[k] \}_{\mathcal{R}}$ via two filtering steps (\textbf{Lines~7 and 8}) and the weighted average step (\textbf{Line~9}). In particular, utilizing Lemma~\ref{lem: dist to curr aux}, we derive the following relationship.
\begin{proposition} \label{prop: aux inexact}
If Assumptions~\ref{asm: robust} and \ref{asm: weight matrices} hold, then
for all $k \in \N$ and $v_i \in \mathcal{R}$, it holds that
\begin{align*}
    \| \z_i[k] - \y[\infty] \| 
    \leq \max_{v_j \in \mathcal{R}} \| \x_j[k] - \y[\infty] \| 
    + 2 \| \y_i[k] - \y[\infty] \|.
\end{align*}
\end{proposition}
By combining the above inequality with the relationship between $\| \z_i [k] - \y [\infty] \|$ and $\| \x_i [k+1] - \y [\infty] \|$ from Lemmas~\ref{lem: dist grad update 1}-\ref{lem: Delta lower bound}, and bounding the second term on the RHS, $\| \y_i[k] - \y[\infty] \|$, using Proposition~\ref{prop: aux convergence}, we obtain a relationship between $\| \x_i [k+1] - \y [\infty] \|$ and $\max_{v_j \in \mathcal{R}} \| \x_j[k] - \y[\infty] \|$.
As a result, we can bound the distance $\max_{v_i \in \mathcal{R}} \| \x_i[k] - \y[\infty] \|$ by a particular bounded sequence defined below.

Define the time-step $k_0 \in \N$ as $k_0 := \max_{\ell \in \{ 1,2,3 \}} k^*_\ell$.
Recall the definition of $\alpha$ and $\beta$ from Proposition~\ref{prop: aux convergence}. Let 
\begin{equation}
    \phi[k_0] = \max_{v_i \in \mathcal{R}} \| \x_i[0] - \y[\infty] \| + 2 \beta \sum_{k=0}^{k_0-1} e^{-\alpha k} + L \sum_{k=0}^{k_0-1} \eta [k], 
    \label{eqn: phi[k0]}
\end{equation}
and define a sequence $\{ \phi[k] \}_{k = k_0}^\infty$ satisfying the update rule 
\begin{equation}
    \phi^2[k+1] = \max \Big\{ ( s^* )^2, 
    \big( \phi[k] + 2 \beta e ^{- \alpha k} \big)^2 - \frac{1}{2} \eta[k] \min_{v_i \in \mathcal{R}} b_i \underline{L}_i \Big\}.
    \label{eqn: phi k+1}
\end{equation}

\begin{lemma} \label{lem: two bounds}
Suppose Assumptions~\ref{asm: convex}-\ref{asm: weight matrices} hold.
For all $k \geq k_0$, it holds that
\begin{equation*}
    \max_{v_i \in \mathcal{R}} \| \x_i[k] - \y[\infty] \| \leq \phi[k].
\end{equation*}
Furthermore, there exists $\bar{\phi} \in \R_{\geq 0}$ such that for all $k \geq k_0$, the sequence $\phi[k]$ can be uniformly bounded as $\phi[k] < \bar{\phi}$.
\end{lemma}

\subsubsection{Convergence Analysis}
\label{subsec: conv analysis}

Finally, we will utilize the following lemma to further analyze the sequence $\{ \phi[k] \}$ defined in \eqref{eqn: phi k+1}.
\begin{lemma}  \label{lem: update imp}
Consider a sequence $\{ \hat{\eta}[k] \}_{k=0}^{\infty} \subset \R_{\geq 0}$ that satisfies $\sum_{k=0}^{\infty} \hat{\eta}[k] = \infty$.
If $\gamma_1 \in \R_{\geq 0}$, $\gamma_2 \in \R_{> 0}$ and $\lambda \in (-1, 1)$, then there is no sequence $\{ u[k] \}_{k=0}^{\infty} \subset \R_{\geq 0}$ that satisfies the update rule
\begin{equation*}
    u^2[k+1] = (u[k] + \gamma_1 \lambda^k)^2 - \gamma_2 \hat{\eta}[k].
\end{equation*}
\end{lemma}

By employing Lemmas~\ref{lem: two bounds} and \ref{lem: update imp}, Proposition~\ref{prop: finite time} demonstrates that any repulsion of the state $\z_i[k]$ from the convergence ball $\B( \y[\infty], s^*)$ due to inconsistency of the estimates of the auxiliary point  (Proposition~\ref{prop: aux convergence} and \ref{prop: aux inexact})
is compensated by the gradient term pulling the state $\x_i[k]$ to the convergence ball. Consequently, the quantity $\phi[k]$ decreases until it does not exceed $s^*$. In other words, the sequence analysis results in 
\begin{equation}
    \max_{v_i \in \mathcal{R}} \| \x_i[k] - \y[\infty] \| \leq \phi[k] \leq s^*
    \label{eqn: finite convergence}
\end{equation}
for a sufficiently large time-step $k$. The crucial finite time convergence result is formally stated as follows.
\begin{proposition} \label{prop: finite time}
Suppose Assumptions~\ref{asm: convex}-\ref{asm: weight matrices} hold.
Then, there exists $K \in \N$ such that for all $v_i \in \mathcal{R}$ and $k \geq K$, we have $\x_i[k] \in \B( \y[\infty], s^*)$.
\end{proposition} 
Since all the prior analyses valid for all $\xi \in \R_{> 0}$ and $\epsilon \in \R_{> 0}$, the convergence result in Theorem~\ref{thm: convergence} follows from taking $\inf_{\xi > 0, \epsilon > 0}$ and $\limsup_k$ to \eqref{eqn: finite convergence}.

%% file: contents/sec-discussion.tex
\section{Discussion}  \label{sec: discussion}

\subsection{Fundamental Limitation}
One would ideally expect an algorithm to provide convergence to the exact minimizer of the sum of the regular agents' functions when there are no Byzantine agents in the network. However, prior works \cite{su2016fault, sundaram2018distributed} have established a fundamental limitation, showing that achieving such a guarantee \emph{is not possible} unless the set of local functions possesses a redundancy property, known as $2F$-redundancy \cite{gupta2021byzantine}. This limitation arises from the strong model of Byzantine attacks considered, where a Byzantine agent can substitute the given local function with a forged function that remains legitimate. Consequently, detecting such suspicious behavior or determining the total number of Byzantine agents $| \mathcal{A} |$ in the network is not possible, as the Byzantine agent can follow the algorithm while influencing the outcome of distributed optimization (as discussed in Remark~\ref{rem: fund limit}). In other words, in settings where Byzantine agents are \emph{potentially} present (i.e., $ F > 0$) and there is no known redundancy among the functions, achieving zero steady state error is \emph{impossible} even when there are no Byzantine agents actually present (i.e., $| \mathcal{A} | = 0$) \cite{su2016fault, sundaram2018distributed}.

Our work, imposing only mild assumptions on the local functions, is constrained by this fundamental limit. Although our approach can recover the distributed subgradient method \cite{nedic2009distributed, nedic2014distributed} when selecting the parameter $F = 0$, in the worst case scenario, there is no way to determine the number of Byzantine agents. As discussed in Remark~\ref{rem: F knowledge}, in practice, we need to choose the parameter $F$ in the design phase, i.e., prior to the execution of the algorithm. Thus, in our work, the parameter $F$ serves as the maximal number of Byzantine agents in a set of neighbors that the designed system can tolerate, providing a convergence guarantee, as stated in Theorem~\ref{thm: convergence}.

It is crucial to acknowledge that the fundamental limit is well-established for distributed optimization problems. However, the question of the dependence of the smallest size of the convergence region on the parameters characterizing the function class remains an important open problem \cite{kuwaran2023geometric}.

\subsection{Redundancy and Guarantees Trade-off}  \label{subsec: tradeoff}
An appropriate notion of network redundancy is necessary for any Byzantine resilient optimization algorithm \cite{sundaram2018distributed}; for both Algorithm~\ref{alg: DMM_filter} and Algorithm~\ref{alg: D_filter}, this is captured by the corresponding robustness conditions in Theorem~\ref{thm: convergence}.
In particular, Algorithm~\ref{alg: DMM_filter} requires the graph to be $((2d + 1)F + 1)$-robust since it implements two filters (a distance-based filter (\textbf{Line 7}) and a min-max filter (\textbf{Line 8})) while Algorithm~\ref{alg: D_filter} requires the graph to only be $(2F + 1)$-robust as a result of only using the distance-based filter.
Since each of these filtering steps discards a set of state vectors, the robustness condition allows the graph to retain some flow of information. Thus, while Algorithm~\ref{alg: DMM_filter} requires significantly stronger conditions on the network topology (i.e., requiring the robustness parameter to scale linearly with the dimension of the functions), it provides the benefit of guaranteeing consensus. Algorithm~\ref{alg: D_filter} only requires the robustness parameter to scale with the number of adversaries in each neighborhood, and thus can be used for optimizing high-dimensional functions with relatively sparse networks, at the cost of losing the guarantee on consensus.

\begin{remark}  \label{rem: redundancy}
The linear dependence of the redundancy requirement on the number of dimensions $d$ is, in fact, typical for resilient vector consensus (e.g., see \cite{tverberg1966generalization, reay1968extension, vaidya2014iterative, xiang2016brief, park2017fault, abbas2020interplay, yan2020safe}); The survey paper \cite[Section~5.3]{pirani2022survey} provides a detailed discussion of papers that require this assumption. Despite such a condition/restriction being ``standard'' in the literature, the linear growth in the number of neighbors with the dimension of the state is undesirable. To address the drawback of requiring high redundancy, we provide Algorithm~\ref{alg: D_filter} which is an alternative solution that does not depend on the number of dimensions $d$; however, in this case, we lose the consensus guarantee unlike Algorithm~\ref{alg: DMM_filter}.
\end{remark}

\subsection{Time Complexity}  \label{subsec: complexity}
Suppose the network is $r$-robust and the number of in-neighbors $| \mathcal{N}_i^{\text{in}} |$ is linearly proportional to $r$ for all $v_i \in \mathcal{V}$. For the distance-based filter (\textbf{Line~7}), each regular agent $v_i \in \mathcal{R}$ computes the $L^2$-norm between its auxiliary state and in-neighbor states and then finds the $F$ agents that attain the maximum value; this procedure  takes $\mathcal{O} (dr)$ operations. On the other hand, for the min-max filter (\textbf{Line~8}), each regular agent $v_i \in \mathcal{R}$ is required to sort the in-neighbor states for each dimension which takes $\mathcal{O} (dr \log r)$ operations. For Algorithms~\ref{alg: DMM_filter} and \ref{alg: D_filter}, the total computational complexities for filtering process are $\Tilde{\mathcal{O}} (d^2)$ and $\Tilde{\mathcal{O}} (d)$, respectively.
Compared to the resilient vector consensus literature \cite{tverberg1966generalization, reay1968extension, vaidya2014iterative, xiang2016brief, park2017fault, abbas2020interplay, yan2020safe}, which requires exponential in the number of dimensions $d$ for computational complexity, our algorithms have significantly lower computation costs.

\subsection{Convergence Ball}  \label{subsec: convergence ball}
In terms of the size of the convergence ball, it is crucial to note that the convergence radius defined in \eqref{def: conv radius} remains independent of the Lipschitz constant $L$, the number of regular agents $|\mathcal{R}|$ (in contrast to the result in \cite{wu2023byzantine}), or the maximum number of neighboring Byzantine agents $F$. Instead, the radius hinges solely on specific characteristics of local functions: the locations of local minimizers (captured by $\Tilde{R}_i$), sensitivity (captured by $\theta_i$), and the size of the set of local minimizers (captured by $\delta_i$). However, the quantity $\Tilde{R}_i$ can be proportional to $\sqrt{d}$ in the worst case as analyzed in \cite{kuwaran2023geometric}. As we will discuss next, remarkably, the sensitivity $\theta_i$ defined in Proposition~\ref{prop: grad angle} is intimately linked to the condition number of a function.

For simplicity, we will omit the agent index subscript $i$ in the subsequent analysis and assume $\x_i^* = \0$. Consider a quadratic function $f(\x) = \frac{1}{2} \| \A \x \|^2$, where $A \in \R^{d \times d}$ is a positive definite matrix. Now, we will examine the quantity $\sup_{\x \neq \0} \angle (\g(\x), \x)$ from Proposition~\ref{prop: grad angle}, where $\g(\x) = \nabla f(\x) = \A^T \A \x$. We aim to demonstrate that $\sec \left( \sup_{\x \neq \0} \angle (\g (\x), \x) \right) \leq (\| \A \| \cdot \| \A^{-1} \|)^2 := \kappa$, where ${\| \; \cdot \; \|}$ denotes the induced matrix norm, and $\kappa$ is the condition number associated with the function $f$ \cite{gutman2021condition}. It is noteworthy that this inequality, with the replacement of $\sup_{\x \neq \0} \angle (\g (\x), \x)$ by $\sup_{\x \neq \x_*} \angle (\g (\x), \x - \x_*)$, holds for the more general case of $f(\x) = \frac{1}{2} \| \A ( \x - \x_* ) + \boldsymbol{b} \|^2$ with $\x_* \in \R^d$ and $\boldsymbol{b} \in \R^d$.

To show such result, we proceed as follows:
\begin{align*}
    \cos \Big( \sup_{\x \neq \0} \angle ( \g (\x), \x ) \Big)
    &= \inf_{\x \neq \0} \big( \cos \angle ( \g (\x), \x ) \big)
    = \inf_{\x \neq \0} \left ( \frac{ \langle \g (\x), \x \rangle }{\| \g (\x) \| \cdot \| \x \|} \right )  
    = \inf_{\x \neq \0} \left ( \frac{\| \A \x \|^2}{\| \x \|^2} \cdot \frac{\| \x \|}{\| \A^T \A \x \|} \right ) \\
    &\geq \bigg ( \inf_{\x \neq \0} \frac{\| \A \x \|}{\| \x \|} \bigg )^2 \bigg ( \sup_{\x \neq \0} \frac{\| \A^T \A \x \|}{\| \x \|} \bigg )^{-1} 
    \geq \big( \| \A^{-1} \| \cdot \| \A \| \big)^{-2}.
\end{align*}
In the last inequality, we utilize the properties that $\inf_{\x \neq \0} \frac{\| \A \x \|}{\| \x \|} = \frac{1}{\| \A^{-1} \|}$ due to the invertibility of $\A$ and $\sup_{\x \neq \0} \frac{\| \A^T \A \x \|}{\| \x \|} = \| \A^T \A \| \leq \| \A \|^2$ due to the sub-multiplicative property of induced matrix norm, and $\| \A^T \| = \| \A \|$.

To get a sense of the convergence region, we consider univariate functions (i.e., the $d = 1$ case). To facilitate the discussion, we denote $\min_{v_i \in \mathcal{R}} x_i^*$ and $\max_{v_i \in \mathcal{R}} x_i^*$ by $\underline{x}$ and $\overline{x}$, respectively. Suppose that the local minimizer $x_i^*$ is unique for all $v_i \in \mathcal{R}$ so that the quantity $\delta_i$ defined in \eqref{def: delta_i} can be chosen arbitrarily close to zero for all $v_i \in \mathcal{R}$. In this case, we have that for all $v_i \in \mathcal{R}$, $\theta_i$ defined in \eqref{def: theta_i} is zero. Therefore, the convergence radius $s^*$ in \eqref{def: conv radius} simplifies to $\max_{v_i \in \mathcal{R}} \tilde{R}_i$ (where $\tilde{R}_i$ defined in \eqref{def: tilde R_i}). 
In the best case, we can have $y[\infty] = \frac{1}{2} ( \underline{x} + \overline{x} )$ which results in the convergence region $[ \underline{x}, \overline{x} ]$ as derived in \cite{sundaram2018distributed}. In the worst case, (assuming numerical error $\epsilon^*$ in \textbf{Line~1} is zero) we can have $y[\infty] = \underline{x}$ or $\overline{x}$ which results in the convergence region $[ 2 \underline{x} - \overline{x}, \overline{x} ]$ or $[ \underline{x}, 2 \overline{x} - \underline{x} ]$, respectively. In such worst case, the region is two times bigger than the region derived in \cite{sundaram2018distributed}. These results are due to our ``radius analysis" which is uniform in all directions from $\y[\infty]$.

\begin{remark}
Regarding the convergence rate, given the general convex (possibly non-smooth) nature of the problem, achieving only sublinear convergence is typical in centralized settings \cite{nesterov2003introductory}. Specifically, the anticipated convergence rate may align with the $\mathcal{O} \big( \frac{\log k}{\sqrt{k}} \big)$ rate observed in \cite{nedic2014distributed} for non-faulty distributed cases. While our current work provides asymptotic analysis due to inherent challenges, our future endeavors aim to explore explicit convergence rates for a broader class of Byzantine-resilient distributed optimization algorithms.
\end{remark}

\subsection{Maximum Tolerance}
Based on the robustness condition for each algorithm and a formula from \cite{guerrero2017formations}, given the number of agents $N$ in the complete graph and number of dimensions for the optimization variables $d$, the upper bound on the number of local Byzantine agents $F$ such that the corresponding guarantees still hold, is as follows:
\begin{itemize}
    \item $F = \left\lfloor \frac{N - 1}{2 (2d + 1)} \right\rfloor$ for Algorithm~\ref{alg: DMM_filter}, and
    \item $F = \left\lfloor \frac{1}{4} (N-1) \right\rfloor$ for Algorithm~\ref{alg: D_filter}.
\end{itemize}
From a practical perspective, the robustness property  demonstrates a natural trade-off for the system designer. A network that has a stronger robustness property can tolerate more adversaries, but can also induce more costs.

\subsection{Importance of Main States Computation}
If we simply implement a resilient consensus protocol on local minimizers similar to the auxiliary states, $\y_i[k]$, computation (in \textbf{Lines~11-12}) and remove the main states, $\x_i[k]$, computation (in \textbf{Lines~7-9}), we would obtain that the states of the regular agents converge to the hyper-rectangle formed by the local minimizers (for resilient component-wise consensus algorithms \cite{saldana2017resilient}), or the convex hull of the local minimizers (for resilient vector consensus algorithms \cite{park2017fault, abbas2022resilient}). Even though using a resilient consensus protocol seems to be a good method for the single dimension case since the resilient distributed optimization algorithm also pushes the states of the regular agents to such sets \cite{su2015byzantine, sundaram2018distributed} (and they are identical in this case), it might not give a desired result for the multi-dimensional case. First, it is possible that the minimizer of the sum lies outside both the hyper-rectangle and convex hull \cite{kuwaran2018location, kuwaran2020set} as shown in Figure~\ref{fig: quantity illus}. Second, using only a resilient consensus protocol, one ignores the gradient information which steers the regular agents' states to the true minimizer. Third, we empirically show in Section~\ref{sec: simulations} that implementing a resilient distributed optimization algorithm (especially Algorithm~\ref{alg: DMM_filter}) usually gives better results (compared to the quality of the solution provided by directly using the auxiliary point, which was obtained by running a resilient consensus protocol on the local minimizers) in terms of both optimality gap and distance to the global minimizer. 

\subsection{Importance of Auxiliary States Computation}
Essentially, when the main states of regular agents are significantly far away from their local minimizers $\x_i^*$, these minimizers tend to form a \emph{cluster} from the perspective of a regular agent $v_i$. In addition, building on prior works \cite{kuwaran2018location, kuwaran2020set, kuwaran2023minimizer}, we know that the true optimal solution $\x^*$ (which is the minimizer of the function sum) cannot be located too far away from this cluster. Thus, the auxiliary states $\y_i$, guaranteed to be inside the cluster (in $L_1$-sense) as shown in Proposition~\ref{prop: aux convergence}, act as valuable references for providing a directional sense to regular agents $v_i$ in their pursuit of the true minimizer $\x^*$. By our design, the distance filter in Algorithm~\ref{alg: DMM_filter} and Algorithm~\ref{alg: D_filter} assumes the role of a guiding mechanism by eliminating extreme states that pull the overall state away from the cluster.

From a technical standpoint, in the multi-dimensional case, relying solely on resilient consensus for the main states $\x_i$ and the update using a subgradient $\g_i$ with respect to the local function $f_i$ may not suffice to ensure a convergence guarantee. In the worst case, resilient consensus could lead to a state further away from the cluster, especially considering that the strength of this divergence due to Byzantine agents can be proportional to $\sqrt{d}$, where $d$ is the problem dimension. Even though following the subgradient $\g_i$ usually mitigates the divergence, it might not be sufficient for guaranteed convergence in such worst cases. Thus, our introduced distance-based filter using a local auxiliary state plays a crucial role in further reducing the severity of the divergence, allowing us to achieve a convergence guarantee under mild assumptions.

%% file: contents/sec-experiment.tex
\section{Numerical Experiment} 
\label{sec: simulations}

We now provide two numerical experiments to illustrate Algorithm~\ref{alg: DMM_filter} and Algorithm~\ref{alg: D_filter}. In the first experiment, we generate quadratic functions for the local objective functions. Using these functions, we demonstrate the performance (e.g., optimality gaps, distances to the global minimizer) of our algorithms. We also compare the optimality gaps of the function value obtained using the states $\boldsymbol{x}_i [k]$ and the value obtained using the auxiliary points $\boldsymbol{y}_i [k]$, and plot the trajectories of the states of a subset of regular nodes.
In the second experiment, we demonstrate the performance of our algorithm on a machine learning task (banknote authentication task). Specifically, we compare the accuracy of the models obtained from our algorithm (resilient distributed model) and that of a centralized model.

\subsection{Synthetic Quadratic Functions}

\vspace{1.5mm}
\noindent \textbf{Preliminary Settings} 
\begin{itemize}
    \item \textbf{Main Parameters:} We set the number of nodes to be $n = 25$ and the dimension of each function to be $d = 2$.
    \item \textbf{Adversary Parameters:} We consider the $F$-local model, and set $F = 2$ for Algorithm~\ref{alg: DMM_filter} and $F = 5$ for Algorithm~\ref{alg: D_filter}.
\end{itemize}
\textbf{Network Settings} 
\begin{itemize}
    \item \textbf{Topology Generation:} We construct an $11$-robust graph on $n = 25$ nodes following the approach from \cite{guerrero2017formations, leblanc2013resilient}. This graph can tolerate up to 2 local adversaries for Algorithm~\ref{alg: DMM_filter}, and up to 5 local adversaries for Algorithm~\ref{alg: D_filter} according to Theorem~\ref{thm: convergence}. Note that the same graph is used to perform numerical experiments for both Algorithms~\ref{alg: DMM_filter} and \ref{alg: D_filter}.
\end{itemize}
\noindent \textbf{Adversaries' Strategy}
\begin{itemize}
    \item \textbf{Adversarial Nodes:} We construct the set of adversarial nodes $\mathcal{A}$ by randomly choosing nodes in $\mathcal{V}$ so that the set of adversarial nodes form a $F$-local set. Note that in general, constructing $\mathcal{A}$ depends on the topology of the network. 
    In our experiment, we have $\mathcal{A} = \{ v_9, v_{16} \}$ for Algorithm~\ref{alg: DMM_filter} and $\mathcal{A} = \{ v_5, v_{11}, v_{12}, v_{17}, v_{22}, v_{24} \}$ for Algorithm~\ref{alg: D_filter}.
    
    \item \textbf{Adversarial Values Transmitted:} Here, we use a sophisticated approach rather than simply choosing the transmitted values at random. 
    Suppose $v_s$ is an adversary node and $v_i$ is a regular node which is an out-neighbor of $v_s$, i.e., $v_s \in \mathcal{N}_i^{\text{in}}$.
    First, consider the state of nodes in the network at time-step $k$. The adversarial node $v_s$ uses an oracle to determine the region in the state space for the regular node $v_i$ in which if the adversarial node selects the transmitted value to be outside the region then the value will be discarded by that regular agent $v_i$. 
    Then, $v_s$ chooses $\x_{s \to i} [k]$ (the forged state sent from $v_s$ to $v_i$ at time $k$) so that it is in the safe region and far from the global minimizer. In this way, the adversaries' values will not be discarded and also try to prevent the regular nodes from getting close to the minimizer.
    Similarly, for the auxiliary point update, the adversarial node $v_s$ uses an oracle to determine the safe region in the auxiliary point's space for the regular node $v_i$. Since the safe region is a hyper-rectangle in general, $v_s$ chooses $\y_{s \to i} [k]$ (the forged estimated auxiliary point sent from $v_s$ to $v_i$ at time $k$) to be near a corner (chosen randomly) of the hyper-rectangle.
\end{itemize}
\textbf{Objective Functions Settings}
\begin{itemize}
    \item \textbf{Local Functions:} For $v_i \in \mathcal{V}$, we set the local objective functions $f_i: \R^d \to \R$ to be
    \begin{equation*}
        f_i(\x) = \frac{1}{2} \x^T \boldsymbol{Q}_i \x + \boldsymbol{b}_i^T \x, 
    \end{equation*}
    where $\boldsymbol{Q}_i \in \mathcal{S}_{d}^{+}$ and $\boldsymbol{b}_i \in \R^d$ are chosen randomly.
    Note that the same local functions are used to perform numerical experiments for both Algorithms~\ref{alg: DMM_filter} and \ref{alg: D_filter}.
    
    \item \textbf{Global Objective Function:} According to our objective \eqref{prob: regular node}, we then have the global objective function $f: \R^d \to \R$ as follows:
    \begin{equation*}
        f(\x) = \frac{1}{| \mathcal{R} |} \Big( \frac{1}{2} \x^T \big( \sum_{v_i \in \mathcal{R}} \boldsymbol{Q}_i \big) \x + \big( \sum_{v_i \in \mathcal{R}} \boldsymbol{b}_i \big)^T \x \Big),
    \end{equation*}
    where the set of regular nodes $\mathcal{R} = \mathcal{V} \setminus \mathcal{A}$. 
\end{itemize}
\textbf{Algorithm Settings}
\begin{itemize}
    \item \textbf{Initialization:} For each regular node $v_i \in \mathcal{R}$, we compute the exact minimizer $\x_i^* = - \boldsymbol{Q}_i^T \boldsymbol{b}_i$ and use it as the initial state and auxiliary point of $v_i$ as suggested in \textbf{Line 1-2} of Algorithm~\ref{alg: DMM_filter}.
    
    \item \textbf{Weights Selection:} For each time-step $k \in \N$ and regular node $v_i \in \mathcal{R}$, we randomly choose the weights $w_{x, ij} [k], w^{(\ell)}_{y, ij} [k]$ so that they follow the description of \textbf{Line~9} and \textbf{Line~12}, and Assumption~\ref{asm: weight matrices}.
    
    \item \textbf{Step-size Selection:} We choose the step-size schedule (in \textbf{Line 11} of Algorithm \ref{alg: DMM_filter}) to be $\eta [k] = \frac{1}{k+1}$.
    
    \item \textbf{Gradient Norm Bound:} We choose the upper bound of the gradient norm to be $L = 10^5$. If the norm exceeds the bound, we scale the gradient down so that its norm is equal to $L$, i.e., 
    \begin{equation*}
        \g_i [k] = \begin{cases}
        \nabla f_i ( \z_i [k] ) \qquad \text{if} \;\; \| \nabla f_i ( \z_i [k] ) \| \leq L, \\
        \frac{L}{ \| \nabla f_i ( \z_i [k] ) \|} \cdot \nabla f_i ( \z_i [k] ) \quad \text{otherwise}.
        \end{cases}
    \end{equation*}
\end{itemize}
\textbf{Simulation Settings and Results}
\begin{itemize}
    \item \textbf{Time Horizon:} We set the time horizon of our simulations to be $K = 300$ (starting from $k = 0$).
    
    \item \textbf{Experiments Detail:} For both Algorithms~\ref{alg: DMM_filter} and \ref{alg: D_filter}, we fix the graph, local functions, and step-size schedule. However, since the set of adversaries are different, the global objective functions, and hence the global minimizers are different. For each algorithm, we run the experiment $10$ times setting the same states initialization across the runs. The results from the runs are different due to the randomness in the adversaries' strategy.
    
    \item \textbf{Performance Metrics:} We examine the performance of our algorithms by considering the optimality gaps (Figure~\ref{fig: optimality_gap}), distances to the global minimizer (Figure~\ref{fig: D2M}), and trajectories of randomly selected regular agents (Figure~\ref{fig: trajectory}).
    
    \item \textbf{Algorithm~\ref{alg: DMM_filter}'s Results:} The lines corresponding to the optimality gap and distance to the global minimizer evaluated using auxiliary points are almost horizontal since the convergence to consensus is very fast.
    However, one can see that the optimality gap and distance to the minimizer obtained from the regular states are significantly smaller than that from the auxiliary points due to the use of gradient information (\textbf{Line~10}) and extreme states filtering (\textbf{Line~8}) in the regular state update. 
    In particular, at $k=300$, the optimality gap and distance to the global minimizer at the regular states' average are only about $0.030$ and $0.206$, respectively.
    Moreover, the state trajectories converge together and stay close to the global minimizer even in the presence of sophisticated adversaries.
    Note that, from our observations, Algorithm~\ref{alg: DMM_filter} yields better results than Algorithm~\ref{alg: D_filter} given the same settings.
    
    \item \textbf{Algorithm~\ref{alg: D_filter}'s Results:} The optimality gaps and distances to the global minimizer evaluated using the states are slightly better than the values obtained using the auxiliary points, and the state trajectories remain reasonably close to the global minimizer showing that the algorithm can tolerate $F = 5$ local adversaries (which is more than Algorithm~\ref{alg: DMM_filter}).
    Interestingly, the state trajectories seem to converge together even though the consensus guarantee is lacking due to the absence of the distance-based filter.
\end{itemize}

\begin{figure}
\centering
\subfloat[The optimality gap evaluated at the average of the regular nodes' states $f ( \bar{\x} ) - f^*$ averaged over $10$ runs (blue line), and the optimality gap evaluated at the average of the regular nodes' auxiliary points $f ( \bar{\y} ) - f^*$ averaged over $10$ runs (red line).]
{\includegraphics[width=.40\textwidth]{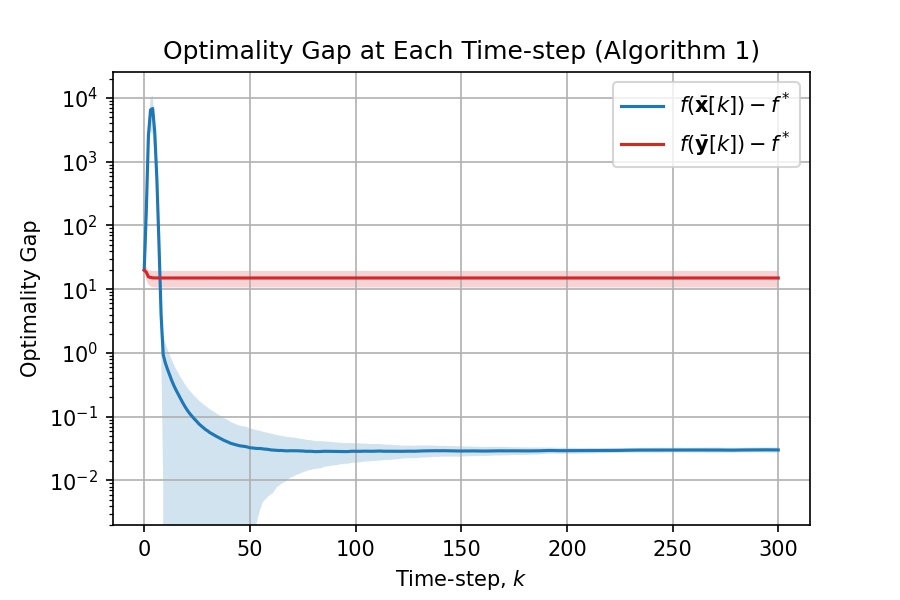}
\includegraphics[width=.40\textwidth]{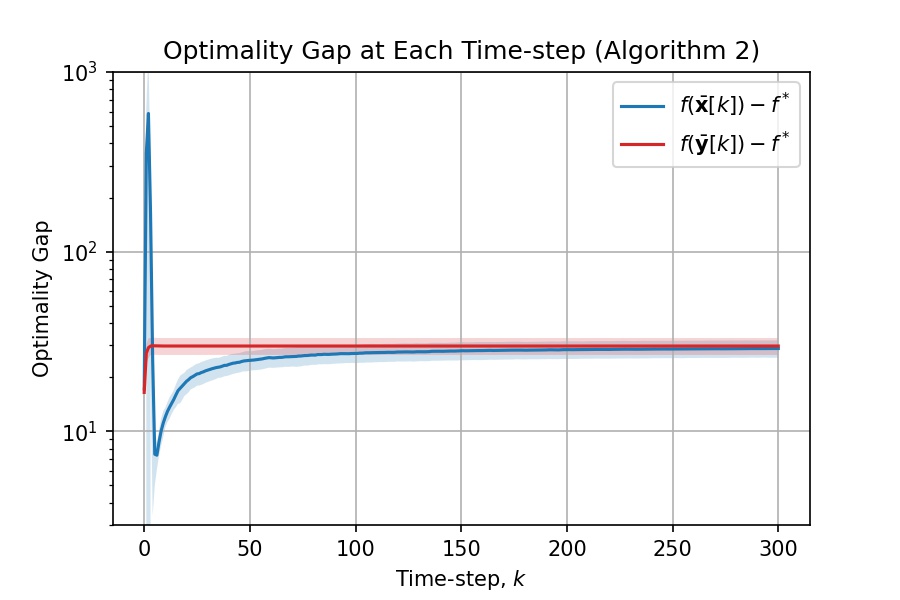} 
\label{fig: optimality_gap}}\;\;
\subfloat[The distance between the average of the regular nodes' states and the global minimizer  $\| \bar{\x} - \x^* \|$ averaged over $10$ runs (blue line), and the distance between the average of the regular nodes' auxiliary points and the global minimizer  $\| \bar{\y} - \x^* \|$ averaged over $10$ runs (red line).]
{\includegraphics[width=.40\textwidth]{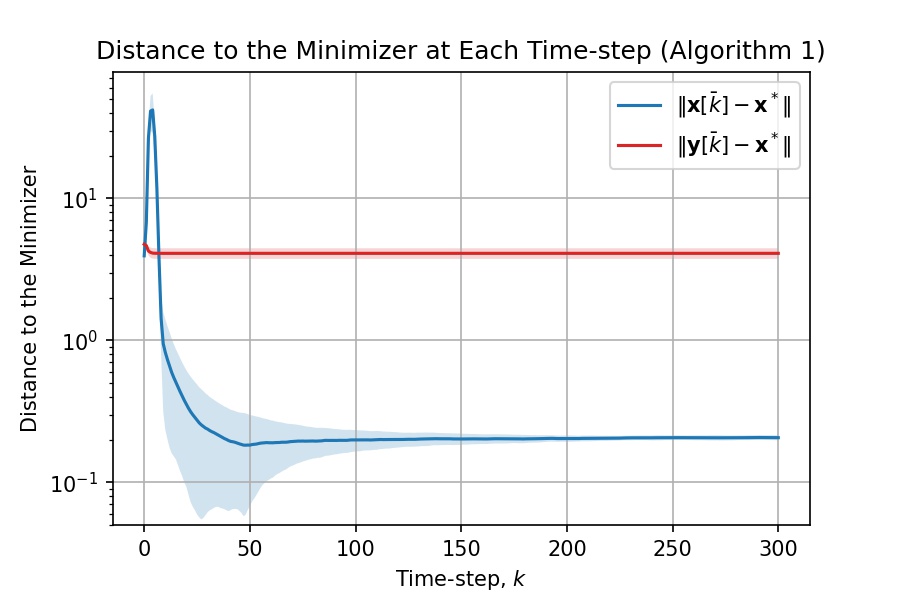}
\includegraphics[width=.40\textwidth]{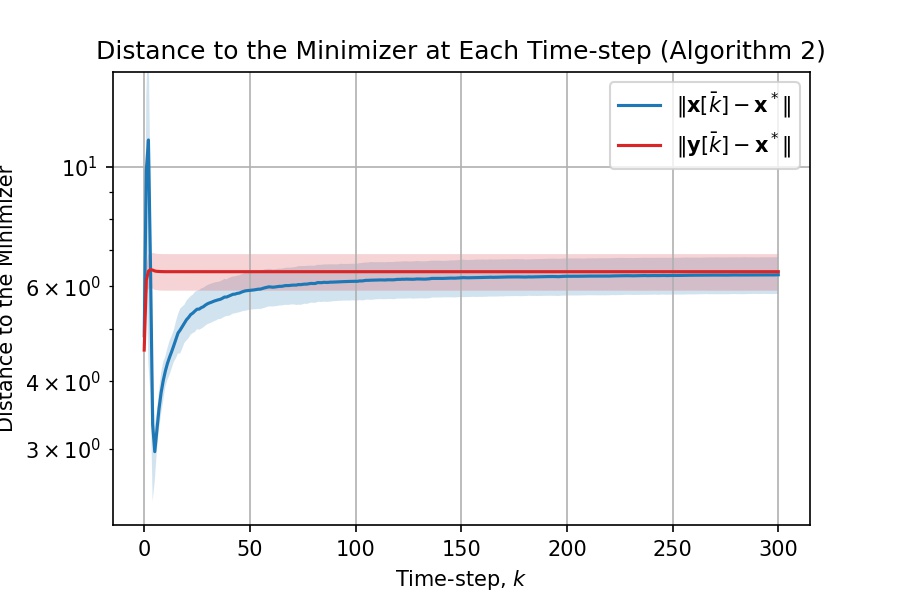}
\label{fig: D2M}}\;\;
\subfloat[The trajectory of the states of a subset of the regular nodes. Different colors of the trajectory represent different regular agents $v_i$ in the network.]
{\includegraphics[width=.40\textwidth]{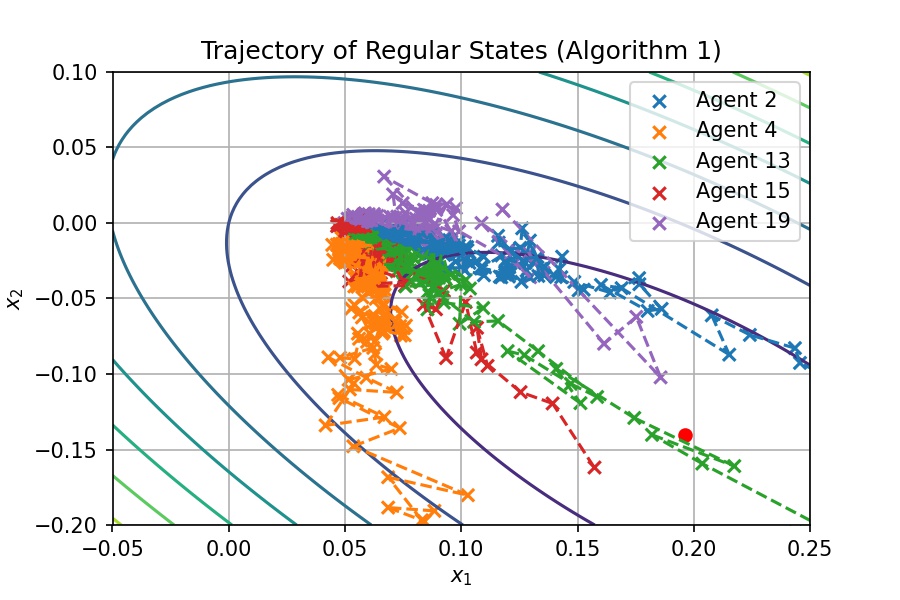}
\includegraphics[width=.40\textwidth]{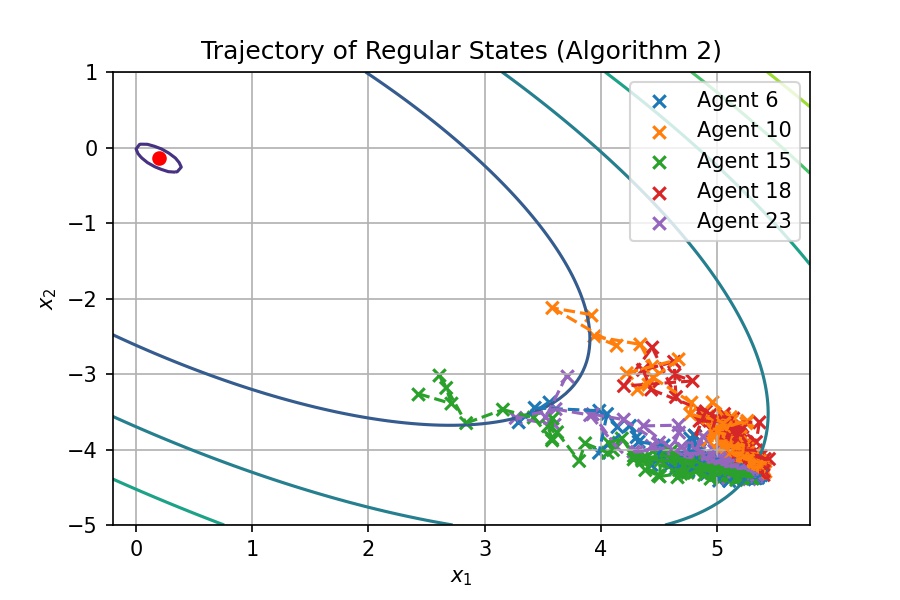}
\label{fig: trajectory}}
\caption{The plots show the results obtained from (left) Algorithm~\ref{alg: DMM_filter} and (right) Algorithm~\ref{alg: D_filter}. In the first four plots, the shaded regions represent +1/-1 standard deviation from the mean. In the last two plots, the contour lines show the level sets of the global objective function (in this case, a quadratic function) and the red dots represent the global minimizer.}
\end{figure}

\subsection{Banknote Authentication using Regularized Logistic Regression}

\noindent \textbf{Dataset Information}\footnote{https://archive.ics.uci.edu/ml/datasets/banknote+authentication}
\begin{itemize}
    \item \textbf{Description:} The data were extracted from images that were taken from genuine and forged banknote-like specimens.
    \item \textbf{Data Points:} The total number of data is 1,372.
    \item \textbf{Features:} The dataset consists of four features: (1) the variance of a wavelet transformed image, (2) the skewness of a wavelet transformed image, (3) the curtosis of a wavelet transformed image, and (4) the entropy of an image.
    \item \textbf{Labels:} There are two classes: `$0$' (genuine) and `$1$' (counterfeit).
\end{itemize}
\textbf{Preliminary Settings} 
\begin{itemize}
    \item \textbf{Main Parameters:} We set the number of nodes to be $n = 50$. Since there are four features, the dimension of the states  is $d = 5$ (one for each feature and the other one for the bias).
    \item \textbf{Adversarial Parameters:} We use the $F$-local model with $F = 2$.
    \item \textbf{Dataset Partitioning:} We randomly partition the dataset into three chunks: 1,000 training data points, 186 validation data points, and 186 test data points. We then distribute the training dataset to the nodes in the network equally. Thus, each node contains $m = 20$ training data points.
\end{itemize}
\textbf{Network and Weights Settings} 

We construct the network and corresponding weight matrix using the same approach as in the synthetic quadratic functions case. \\
\textbf{Adversaries' Strategy}

We choose the set of adversarial nodes $\mathcal{A}$ and adversarial values transmission strategy using the same method as in the synthetic quadratic functions case.

\noindent \textbf{Objective Functions Settings}
\begin{itemize}
    \item \textbf{Notations:} Let $\x_{ij} \in \R^{d-1}$ be the feature vector of the $j$-th data points at node $v_i \in \mathcal{V}$, and $Y_{ij} \in \{0 ,1 \}$ be the corresponding label. We let $\tilde{\x}_{ij} = \begin{bmatrix} \x_{ij}^T & 1 \end{bmatrix}^T$ to account for the bias term.
    \item \textbf{Local Functions:} Since this is a classification task, we choose the logistic regression model with $L_2$-regularization in which its loss function is strongly convex. For $v_i \in \mathcal{V}$, we set the local objective functions $f_i: \R^d \to \R$ to be
    \begin{equation*}
        f_i (\W) = | \mathcal{R} | \sum_{j=1}^{m} \log \big( \exp ( - Y_{ij} \tilde{\x}_{ij}^T \W ) + 1 \big) + \frac{\varsigma}{2} \| \W \|^2,
    \end{equation*}
    where the set of regular nodes $\mathcal{R} = \mathcal{V} \setminus \mathcal{A}$ and $\varsigma \in \R_{> 0}$ is the regularization parameter which will be chosen later.
    \item \textbf{Global Objective Function:} According to our objective \eqref{prob: regular node}, we then have the global objective function $f: \R^d \to \R$ as follows:
    \begin{equation*}
        f( \W ) = \sum_{v_i \in \mathcal{R}} \sum_{j=1}^{m} \log \big( \exp ( - Y_{ij} \tilde{\x}_{ij}^T \W ) + 1 \big) + \frac{\varsigma}{2} \| \W \|^2. 
    \end{equation*}
    \item \textbf{Regularization Parameter Selection:} We consider $\varsigma \in \{ 10^{-4}, 10^{-3}, \ldots, 10^{5} \}$. We train our (centralized) logistic model using the global objective function above for each value of $\varsigma$ and then we select the value of $\varsigma$ that gives the best validation accuracy. 
\end{itemize}
\noindent \textbf{Algorithm Settings}
\begin{itemize}
    \item \textbf{Initialization:} As suggested in \textbf{Line 1} of Algorithm~\ref{alg: DMM_filter}, we numerically find the minimizer of the local functions using the default optimizer of \texttt{sklearn.linear\_model.LogisticRegression}.
    Then, we use the minimizer of each regular node to be the initial state and auxiliary point as in \textbf{Line 2}.
\end{itemize}

The methodology of step-size selection and gradient norm bound is the same as in the synthetic quadratic functions case.

\noindent \textbf{Simulation Settings and Results}
\begin{itemize}
    \item \textbf{Benchmark:} We evaluate the performance (accuracy) of the (centralized) logistic model with the selected regularization parameter, $\varsigma$. 
    \item \textbf{Time Horizon:} We set the time horizon of our simulations of our distributed algorithm to be $K = 200$ (starting from $k = 0$).
    \item \textbf{Simulation:} We run the simulations of Algorithm \ref{alg: DMM_filter} by varying the parameter $\eta_0$ from $-2$ to $4$ with increasing step of $1$. 
    We evaluate the performance of each model (i.e., each $\eta_0$) by considering the accuracy obtained by using the state $\bar{ \W } [K] = \frac{1}{| \mathcal{R} |} \sum_{v_i \in \mathcal{R}} \W_i [K]$ for each $\eta_0$ and the validation data.
    Then, we select the parameter $\eta_0$ which provides the best accuracy. Finally, with the selected value of $\eta_0$, we evaluate the performance (accuracy) of the corresponding model with the test data. 
    \item \textbf{Result:} We repeat the whole process $5$ times. In other words, each run uses different realization of data partitioning (hence, different local functions and global function), network topology, and adversaries set. The result of each run is shown in Table~\ref{table: Logistic}. 
    The first three rows show the adversaries set, regularization parameter and step-size parameter of each run. The next (resp. last) three rows show the training (resp. test) accuracy of the centralized model, distributed model evaluated at $\bar{ \W } [K] = \frac{1}{| \mathcal{R} |} \sum_{v_i \in \mathcal{R}} \W_i [K]$, and the minimum accuracy among the local model of regular nodes evaluated at its own state $\W_i [K]$.
    We can see that despite the presence of adversaries with sophisticated behavior, the performance of our algorithm is just slightly lower than the centralized model's performance for this task.
\end{itemize}

\begin{table} 
\centering
\begin{tabular}{ || c | c c c c c || }
\hline 
 & 1st Run & 2nd Run & 3rd Run & 4th Run & 5th Run \\ [0.5ex]
\hline \hline
$v_i \in \mathcal{A}$ & $9, 23$ & $9, 15$ & $10, 11$ & $5, 29$ & $24, 47$ \\
$\varsigma$ & 1.0 & 1.0 & 10 & 1.0 & 1.0 \\ 
$\eta_0$ & 1 & 1 & 2 & 4 & 1 \\
\hline
Train (C) & 99.40 & 99.20 & 99.00 & 98.90 & 99.10 \\
Train (D) & 98.10 & 97.90 & 97.70 & 98.30 & 98.00 \\
Train (MIN) & 97.80 & 97.60 & 97.50 & 98.30 & 97.70 \\
\hline
Test (C) & 99.46 & 98.39 & 99.46 & 99.46 & 98.92 \\
Test (D) & 97.85 & 95.70 & 98.92 & 97.85 & 98.39 \\
Test (MIN) & 97.85 & 95.70 & 98.92 & 97.85 & 97.85 \\
\hline
\end{tabular}
\vspace*{2mm}
\caption{Training/Test Accuracy of Centralized (C), Distributed (D) Models and Minimum among Regular Agents' Models (MIN) for each Run of Banknote Authentication Task}
\label{table: Logistic}
\end{table}

%% file: contents/sec-conclusion.tex
\section{Conclusion and Future work} 
\label{sec: conclusion}

In this paper, we considered the distributed optimization problem in the presence of Byzantine agents. We developed two resilient distributed optimization algorithms for multi-dimensional functions. The key improvement over our previous work in \cite{kuwaran2020byzantine} is that the algorithms proposed in this paper do not require a fixed auxiliary point to be computed in advance (which will not happen under finite time in general). Our algorithms have low complexity and each regular node only needs local information to execute the steps.
Algorithm~\ref{alg: DMM_filter} (with the min-max state filter), which requires more network redundancy, guarantees that the regular states can asymptotically reach consensus and enter a bounded region that contains the global minimizer, irrespective of the actions of Byzantine agents. On the other hand, Algorithm~\ref{alg: D_filter} (without the min-max filter) has a more relaxed condition on the network topology and can guarantee  asymptotic convergence to the same region, but cannot guarantee consensus. For both algorithms, we explicitly characterized the size of the convergence region, and showed through simulations that Algorithm~\ref{alg: DMM_filter} appears to yield results that are closer to optimal, as compared to Algorithm~\ref{alg: D_filter}.

As noted earlier, the consensus guarantee for Algorithm~\ref{alg: DMM_filter} requires linear scaling of network robustness with the dimension of the local functions, which can be limiting in practice. This seems to be a common challenge for resilient consensus-based algorithms in systems with multi-dimensional states, e.g., \cite{yan2020resilient, gupta2021byzantine, abbas2022resilient}. Finding a relaxed condition on the network topology for high-dimensional resilient distributed optimization problems (with guaranteed consensus) would be a rich area for future research.

%% file: contents/sec-supp_lemma.tex
\section{Additional Lemma} 
\label{sec: lemma}

We provide a lemma which is utilized in the proof of Theorem~\ref{thm: true_sol} and Lemma~\ref{lem: grad update}.
\begin{lemma}  \label{lem: max_angle}
For given $\hat{\x} \in \R^d$ and $R \geq 0$, if $\x \notin \B(\hat{\x}, R)$ then
\begin{align*}
    \max_{\y \in \B(\hat{\x}, R)}  
    \angle (\x - \y, \; \x - \hat{\x})  
    = \arcsin \Big( \frac{ R }{ \| \x - \hat{\x} \|} \Big). 
\end{align*}
\end{lemma} 

\begin{proof}
Since the angle is measured with respect to the vector $\x - \hat{\x}$, consider any 2-D planes passed through the center $\hat{\x}$ and the point $\x$. Since the planes pass through $\hat{\x}$, the intersections between of the ball $\B( \hat{\x}, R )$ and the planes are great circles of radius $R$. Thus, all of the intersections generated from each plane are identical and we can consider the angle using a great circle instead of the ball. From geometry, the maximum angle $\phi = \angle ( \x - \y^*, \; \x - \hat{\x} )$ only occurs when the ray starting from the point $\x$ touches the circle at point $\y^*$. Therefore, $\angle ( \hat{\x} - \y^*, \x - \y^*) = \frac{\pi}{2}$ and $\| \hat{\x} - \y^* \| = R$. We have
\begin{align*}
    \sin \phi 
    = \frac{ \| \hat{\x} - \y^* \| }{ \| \hat{\x} - \x \| } 
    = \frac{R}{\| \hat{\x} - \x \|},
\end{align*}
and the result follows.
\end{proof}

%% file: contents/sec-supp_proof_aux.tex
\section{Proof of Proposition~\ref{prop: aux convergence}} 
\label{sec: proof of aux prop}

\begin{proof}[Proof of Proposition~\ref{prop: aux convergence}]
For any $\mathcal{S} \subseteq \mathcal{V}$, $\zeta \in \mathbb{R}$ and $\Bar{k}, k  \in \mathbb{N}$ with $\Bar{k} \geq k$, define the sets
\begin{align*}
    &\mathcal{J}^{(\ell)}_M (\mathcal{S}, \Bar{k}, k, \zeta) := \{ v_i \in \mathcal{S} :  y^{(\ell)}_i[\Bar{k}] > M^{(\ell)}[k] - \zeta \},  \\
    &\mathcal{J}^{(\ell)}_m (\mathcal{S}, \Bar{k}, k, \zeta) := \{ v_i \in \mathcal{S} :  y^{(\ell)}_i[\Bar{k}] < m^{(\ell)}[k] + \zeta \}.
\end{align*}
Consider a fixed $\ell \in \{ 1, 2, \ldots, d \}$ and any time-step $k \in \mathbb{N}$. Define $\zeta^{(\ell)}_0 = \frac{1}{2} D^{(\ell)}[k]$. Note that the set $\mathcal{J}^{(\ell)}_M (\mathcal{V}, k, k, \zeta^{(\ell)}_0 ) \cap \mathcal{J}^{(\ell)}_m (\mathcal{V}, k, k, \zeta^{(\ell)}_0 ) = \emptyset$.

By the definition of these sets, when $D^{(\ell)}[k] > 0$, the sets $\mathcal{J}^{(\ell)}_M (\mathcal{R}, k, k, \zeta^{(\ell)}_0) \neq \emptyset$ and $\mathcal{J}^{(\ell)}_m (\mathcal{R}, k, k, \zeta^{(\ell)}_0) \neq \emptyset$. Since the graph is $(2F + 1)$-robust, at least one of $\mathcal{J}^{(\ell)}_M (\mathcal{R}, k, k, \zeta^{(\ell)}_0 )$ or $\mathcal{J}^{(\ell)}_m (\mathcal{R}, k, k, \zeta^{(\ell)}_0 )$ is $(2F+1)$-reachable which means that at least one of them contains a vertex that has at least $2F+1$ in-neighbors from outside it. 

If such a node $v_i$ is in $\mathcal{J}^{(\ell)}_M (\mathcal{R}, k, k, \zeta^{(\ell)}_0)$, we claim that in the update, $v_i$ cannot use the values strictly greater than $M^{(\ell)}[k]$ and it uses at least one value from $\mathcal{V} \setminus \mathcal{J}^{(\ell)}_M (\mathcal{V}, k, k, \zeta^{(\ell)}_0)$.
To show the first claim, note that the nodes that possess the value (in $\ell$-component) greater than $M^{(\ell)}[k]$ must be Byzantine agents by the definition of $M^{(\ell)}[k]$. Since the regular node $v_i$ discards up to $F$-highest values and there are at most $F$ Byzantine in-neighbors, the Byzantine agents that hold the value greater than $M^{(\ell)}[k]$ must be discarded. To show the second claim,
let $\mathcal{S}^{(\ell)}_1 [k] = \mathcal{J}^{(\ell)}_M (\mathcal{A}, k, k, \zeta^{(\ell)}_0)$ and $\mathcal{S}^{(\ell)}_2 [k] = \mathcal{V} \setminus \mathcal{J}^{(\ell)}_M (\mathcal{V}, k, k, \zeta^{(\ell)}_0)$ to simplify the notation. We have
\begin{itemize}
    \item $\mathcal{V} \setminus \mathcal{J}^{(\ell)}_M (\mathcal{R}, k, k, \zeta^{(\ell)}_0) = \mathcal{S}^{(\ell)}_1 [k] \cup \mathcal{S}^{(\ell)}_2 [k]$, and
    \item $\mathcal{S}^{(\ell)}_1 [k] \cap \mathcal{S}^{(\ell)}_2 [k] = \emptyset$.
\end{itemize}
From Assumption~\ref{asm: robust}, we have $| \mathcal{S}_1^{(\ell)} [k] \cap \mathcal{N}_i^{\text{in}} | \leq F$. Applying $(2 F + 1)$-reachable property of $v_i$ and two above properties, we obtain that $| \mathcal{S}_2^{(\ell)} [k] \cap \mathcal{N}_i^{\text{in}} | \geq F + 1$.
Let $\bar{\mathcal{V}}_i^{(\ell)} [k] \subseteq \mathcal{N}_i^{\text{in}}$ be the set of nodes that $v_i \in \mathcal{R}$ discards their values in dimension $\ell$ at time-step $k$.
From the fact that $v_i \in \mathcal{J}^{(\ell)}_M (\mathcal{R}, k, k, \zeta^{(\ell)}_0)$, and \textbf{Line~11} of Algorithm \ref{alg: DMM_filter}, we know that $| \mathcal{S}^{(\ell)}_2 [k] \cap \bar{\mathcal{V}}_i^{(\ell)} [k] | \leq F$. 
Combining this with the former statement, we can conclude that $v_i$ uses at least one value from $\mathcal{S}^{(\ell)}_2 [k]$ in its update, i.e., $(\mathcal{S}_2^{(\ell)}[k] \cap \mathcal{N}_i^{\text{in}}) \setminus \bar{\mathcal{V}}_i^{(\ell)} [k] \neq \emptyset$.

Consider the auxiliary point update rule \eqref{eqn: weight average y} (in \textbf{Line~12} of Algorithm~\ref{alg: DMM_filter}). We can rewrite the update as 
\begin{equation*}
    y^{(\ell)}_i[k+1] = 
    \sum_{v_j \in (\mathcal{S}_2^{(\ell)}[k] \cap \mathcal{N}_i^{\text{in}}) \setminus \bar{\mathcal{V}}_i^{(\ell)} [k] } w^{(\ell)}_{y, ij}[k] \; y^{(\ell)}_j[k] \;\;
    + \sum_{v_j \in \big( v_i \cup ( \mathcal{J}_M (\mathcal{V}, k, k, \zeta_0^{(\ell)} ) \cap \mathcal{N}_i^{\text{in}}) \big) \setminus \bar{\mathcal{V}}_i^{(\ell)} [k] } w^{(\ell)}_{y,ij}[k] \; y^{(\ell)}_j[k].
\end{equation*}
Since $y^{(\ell)}_j[k]$ on the first and second terms on the RHS are upper bounded by $M^{(\ell)}[k] - \zeta^{(\ell)}_0$ and $M^{(\ell)}[k]$, respectively, and the non-zero weights $w^{(\ell)}_{y,ij}[k]$ are lower bounded by the constant $\omega$ (Assumption~\ref{asm: weight matrices}), the value of this node at the next time-step is upper bounded as
\begin{equation*}
    y^{(\ell)}_i[k+1] \leq \omega ( M^{(\ell)}[k] - \zeta^{(\ell)}_0 ) + (1-\omega) M^{(\ell)} [k] 
    = M^{(\ell)} [k] - \omega \zeta^{(\ell)}_0 .
\end{equation*}
Note that the above bound is applicable to any node that is in $\mathcal{R} \setminus \mathcal{J}^{(\ell)}_M (\mathcal{V}, k, k, \zeta^{(\ell)}_0)$, since such a node will use its own value in its update. Similarly, if there is a node $v_j \in \mathcal{J}^{(\ell)}_m (\mathcal{R}, k, k, \zeta^{(\ell)}_0 )$ that uses the value of a node outside that set, then $y^{(\ell)}_j [k+1] \geq m^{(\ell)} [k] + \omega \zeta^{(\ell)}_0$. This bound is also applicable to any node that is in $\mathcal{R} \setminus \mathcal{J}^{(\ell)}_m (\mathcal{V}, k, k, \zeta^{(\ell)}_0)$.

Now, define the quantity $\zeta^{(\ell)}_1 = \omega \zeta^{(\ell)}_0$. We have that the set $\mathcal{J}^{(\ell)}_M (\mathcal{V}, k+1, k, \zeta^{(\ell)}_1) \cap \mathcal{J}^{(\ell)}_m (\mathcal{V}, k+1, k, \zeta^{(\ell)}_1) = \emptyset$. Furthermore, by the bounds provided above, we see that at least one of the following must be true:
\begin{align*}
    | \mathcal{J}^{(\ell)}_M (\mathcal{R}, k+1, k, \zeta^{(\ell)}_1) | &< | \mathcal{J}^{(\ell)}_M (\mathcal{R}, k, k, \zeta^{(\ell)}_0) |, \quad \text{or} \\
    | \mathcal{J}^{(\ell)}_m (\mathcal{R}, k+1, k, \zeta^{(\ell)}_1) | &< | \mathcal{J}^{(\ell)}_m (\mathcal{R}, k, k, \zeta^{(\ell)}_0) |.
\end{align*}
If $\mathcal{J}^{(\ell)}_M (\mathcal{R}, k+1, k, \zeta^{(\ell)}_1) \neq \emptyset$ and $\mathcal{J}^{(\ell)}_m (\mathcal{R}, k+1, k, \zeta^{(\ell)}_1) \neq \emptyset$, then again by the fact that the graph is $(2F+1)$-robust, there is at least one node in one of these sets that has at least $2F+1$ in-neighbors outside from the set. Suppose $v_i \in \mathcal{J}^{(\ell)}_M (\mathcal{R}, k+1, k, \zeta^{(\ell)}_1)$ is such a node. Then, $v_i$ cannot use the values strictly greater than $M^{(\ell)}[k+1]$ and it uses at least one value from $\mathcal{V} \setminus \mathcal{J}^{(\ell)}_M (\mathcal{V}, k+1, k, \zeta^{(\ell)}_1)$.
Since at time-step $k$, all regular nodes cannot use values that are strictly greater than $M^{(\ell)}[k]$ in the update, we have that $M^{(\ell)}[k+1] \leq M^{(\ell)}[k]$.
Therefore, the value of node $v_i$ at the next time-step is upper bounded as
\begin{equation*}
    y_i^{(\ell)}[k+2] \leq \omega ( M^{(\ell)}[k] - \zeta^{(\ell)}_1 ) + (1- \omega) M^{(\ell)}[k+1]  
    \leq M^{(\ell)}[k] - \omega^2 \zeta^{(\ell)}_0.
\end{equation*}
Again, this upper bound also holds for any regular node that is in $\mathcal{R} \setminus \mathcal{J}^{(\ell)}_M (\mathcal{V}, k+1, k, \zeta^{(\ell)}_1)$. Similarly, if there is a node $v_j \in \mathcal{J}^{(\ell)}_m (\mathcal{R}, k+1, k, \zeta^{(\ell)}_1)$ that has $2F+1$ in-neighbors from outside that set then $y_j^{(\ell)}[k+2] \geq m^{(\ell)}[k] + \omega^2 \zeta^{(\ell)}_0$. This bound also holds for any regular node that is not in the set $\mathcal{R} \setminus  \mathcal{J}^{(\ell)}_m (\mathcal{V}, k+1, k, \zeta^{(\ell)}_1)$.

We continue in this manner by defining $\zeta^{(\ell)}_s = \omega^s \zeta^{(\ell)}_0$ for $s \in \mathbb{N}$. At each time step $k + s$, if both $\mathcal{J}^{(\ell)}_M (\mathcal{R}, k+s, k, \zeta^{(\ell)}_s) \neq \emptyset$ and $\mathcal{J}^{(\ell)}_m (\mathcal{R}, k+s, k, \zeta^{(\ell)}_s) \neq \emptyset$ then at least one of these sets will shrink in the next time-step. If either of the sets is empty, then it will stay empty at the next time-step, since every regular node outside that set will have its value upper bounded by $M^{(\ell)}[k] - \zeta^{(\ell)}_s$ or lower bounded by $m^{(\ell)}[k] + \zeta^{(\ell)}_s$. After $| \mathcal{R} | - 1$ time-steps, at least one of the sets $\mathcal{J}^{(\ell)}_M (\mathcal{R}, k+ | \mathcal{R} | - 1, k, \zeta^{(\ell)}_{| \mathcal{R} | - 1})$ or $\mathcal{J}^{(\ell)}_m (\mathcal{R}, k+ | \mathcal{R} | - 1, k, \zeta^{(\ell)}_{| \mathcal{R} | - 1})$ must be empty since the sets $\mathcal{J}^{(\ell)}_M (\mathcal{R}, k, k, \zeta^{(\ell)}_0)$ and $\mathcal{J}^{(\ell)}_m (\mathcal{R}, k, k, \zeta^{(\ell)}_0)$ can contain at most $\mathcal{R} - 1$ regular nodes. Suppose the former set is empty; this means that
\begin{equation*}
    M^{(\ell)}[k + |\mathcal{R}| - 1] \leq M^{(\ell)}[k] - \zeta^{(\ell)}_{|\mathcal{R}| - 1}.
\end{equation*}
Since $m^{(\ell)}[k + |\mathcal{R}| - 1] \geq m^{(\ell)}[k]$, we obtain
\begin{equation}
    D^{(\ell)}[k + |\mathcal{R}| - 1] 
    \leq D^{(\ell)}[k] - \zeta^{(\ell)}_{|\mathcal{R}| - 1} 
    = \Big( 1 - \frac{\omega^{|\mathcal{R}| - 1}}{2} \Big) D^{(\ell)}[k] 
    = \gamma D^{(\ell)}[k]. 
    \label{eqn: D ineq}
\end{equation}
The first equality comes from the fact that $\zeta^{(\ell)}_s = \omega^s \zeta^{(\ell)}_0$ and $\zeta^{(\ell)}_0 = \frac{1}{2} D^{(\ell)}[k]$.
The same expression as \eqref{eqn: D ineq} arises if the set $\mathcal{J}^{(\ell)}_m (\mathcal{R}, k+ | \mathcal{R} | - 1, k, \zeta^{(\ell)}_{| \mathcal{R} | - 1}) = \emptyset$.

Using the fact that
\begin{equation}
    \big[ m^{(\ell)}[k+1], \; M^{(\ell)}[k+1] \big] \subseteq \big[ m^{(\ell)}[k], \; M^{(\ell)}[k] \big]  
    \label{eqn: D contraction}
\end{equation}
for all $k \in \mathbb{N}$ and the inequality \eqref{eqn: D ineq}, we can conclude that for all $v_i \in \mathcal{R}$, $\lim_{k \to \infty} y_i^{(\ell)}[k] = y^{(\ell)}[\infty]$ exists and for all $k$, we have
\begin{equation}
    y^{(\ell)}[\infty] \in \big[ m^{(\ell)}[k], \; M^{(\ell)}[k] \big].  
    \label{eqn: y infty}
\end{equation}
This completes the first part of the proof.

For the second part, let consider the quantity $D^{(\ell)}[k]$ as follows. For all $k \in \mathbb{N}$, we can write
\begin{equation}
    D^{(\ell)}[k] 
    \leq D^{(\ell)} \bigg[ \Big\lfloor \frac{k}{| \mathcal{R} | - 1} \Big\rfloor ( | \mathcal{R} | - 1 ) \bigg]
    \leq \gamma^{ \big\lfloor \frac{k}{| \mathcal{R} | - 1} \big\rfloor } D^{(\ell)}[0]
    < \gamma^{  \frac{k}{| \mathcal{R} | - 1} - 1 } D^{(\ell)}[0].  
    \label{eqn: D ineq 2}
\end{equation}
The first inequality is obtained by using $x \geq \lfloor x \rfloor$ and \eqref{eqn: D contraction}. To obtain the second inequality, we apply the inequality \eqref{eqn: D ineq} $\big\lfloor \frac{k}{| \mathcal{R} | - 1} \big\rfloor$ times. The last inequality comes from the fact that $\gamma < 1$ and $\lfloor x \rfloor > x-1$ implies $\gamma^{\lfloor x \rfloor} < \gamma^{x-1}$. From \eqref{eqn: y infty} and \eqref{eqn: D ineq 2}, for all $v_i \in \mathcal{R}$, we have
\begin{equation}
    | y_i^{(\ell)}[k] - y^{(\ell)}[\infty] | 
    \leq D^{(\ell)}[k] 
    < \gamma^{  \frac{k}{| \mathcal{R} | - 1} - 1 } D^{(\ell)}[0].  
    \label{eqn: y^l contract}
\end{equation}
Since the inequality \eqref{eqn: y^l contract} holds for all $\ell \in \{ 1,2,\ldots, d \}$, we have
\begin{equation*}
    \| \y_i[k] - \y[\infty] \|^2 
    = \sum_{\ell = 1}^d | y_i^{(\ell)}[k] - y^{(\ell)}[\infty] |^2 
    < \gamma^{ 2 \big( \frac{k}{| \mathcal{R} | - 1} - 1 \big) } \sum_{\ell = 1}^d \big( D^{(\ell)}[0] \big)^2.
\end{equation*}
Taking square root of both sides yields
\begin{equation*}
    \| \y_i[k] - \y_i[\infty] \| 
    < \gamma^{  \frac{k}{| \mathcal{R} | - 1} - 1  } \| \boldsymbol{D}[0] \| 
    = \frac{1}{\gamma} \| \boldsymbol{D}[0] \| \; e^{- \frac{1}{ | \mathcal{R} | - 1} \log ( \frac{1}{\gamma} ) \; k},
\end{equation*}
which completes the proof.
\end{proof}

%% file: contents/sec-supp_proof_grad.tex
\section{Proof of Proposition~\ref{prop: grad angle}}
\label{sec: proof_grad_angle}

\begin{proof}[Proof of Proposition~\ref{prop: grad angle}]
Consider a regular agent $v_i \in \mathcal{R}$.
From Assumption~\ref{asm: convex}, for all $\x$, $\y \in \R^d$, we have $f_i (\y) \geq f_i (\x) + \langle \tilde{\g}_i (\x), \; \y - \x \rangle$, where $\tilde{\g}_i (\x) \in \partial f_i (\x)$. 
Substitute a minimizer $\x_i^*$ of the function $f_i$ into the variable $\y$ to get
\begin{equation}
    - \langle \tilde{\g}_i (\x), \; \x_i^* - \x \rangle \geq f_i (\x) - f_i (\x_i^*).  
    \label{eqn: convex ineq}
\end{equation}
Let $\hat{\theta}_i (\x) = \angle ( \tilde{\g}_i (\x), \; \x - \x_i^* )$.
The inequality \eqref{eqn: convex ineq} becomes
\begin{equation*}
    \| \tilde{\g}_i (\x) \| \; \| \x_i^* - \x \| \cos \hat{\theta}_i (\x) \geq f_i (\x) - f_i (\x_i^*).
\end{equation*}
Fix $\epsilon \in \R_{>0}$, and  suppose that $\x \notin \mathcal{C}_i (\epsilon)$. From Assumption~\ref{asm: gradient_bound}, applying $\| \tilde{\g}_i (\x) \| \leq L$, we have
\begin{equation}
    \cos \hat{\theta}_i (\x) 
    \geq \frac{f_i (\x) - f_i (\x^*)}{L \; \| \x_i^* - \x \|}.  
    \label{eqn: cos theta}
\end{equation}
Let $\Tilde{\x}_i \in \R^d$ be the point on the line connecting $\x_i^*$ and $\x$ such that $f_i ( \Tilde{\x}_i ) = f_i ( \x_i^* ) + \epsilon$. We can rewrite the point $\x$ as
\begin{align*}
    \x = \x_i^* + t ( \Tilde{\x}_i - \x_i^* )  \quad \text{where} \quad t = \frac{ \| \x - \x_i^* \| }{ \| \Tilde{\x}_i - \x_i^* \| } \geq 1.
\end{align*}
Consider the term on the RHS of \eqref{eqn: cos theta}. Since $\Tilde{\x}_i \in \mathcal{C}_i (\epsilon)$, and \eqref{def: delta_i} holds, we have
\begin{align}
    \frac{ f_i (\x) - f_i (\x_i^*) }{ \| \x - \x_i^* \| } 
    &= \frac{f_i ( \x_i^* + t ( \Tilde{\x}_i - \x_i^*) ) - f_i (\x_i^*) }{ t \| \Tilde{\x}_i - \x_i^* \|}  \nonumber \\
    &\geq  \frac{f_i( \x_i^* + t ( \Tilde{\x}_i - \x_i^*) ) - f_i (\x_i^*) }{ t \cdot \max_{\y \in \mathcal{C}_i (\epsilon)} \| \y - \x_i^* \|} \nonumber \\
    &\geq  \frac{f_i ( \x_i^* + t ( \Tilde{\x}_i - \x_i^*) ) - f_i (\x_i^*) }{ t \; \delta_i (\epsilon)}.  \label{eqn: t increase}
\end{align}
Since the quantity $\frac{ f_i( \x_i^* + t ( \Tilde{\x}_i - \x_i^*) ) - f_i (\x_i^*) }{ t }$
is non-decreasing in $t \in [1, \infty)$ \cite[Lemma~2.80]{mordukhovich2013easy}, the inequality \eqref{eqn: t increase} becomes
\begin{equation}
    \frac{ f_i (\x) - f_i (\x_i^*) }{ \| \x - \x_i^* \| } 
    \geq \frac{ f_i (\Tilde{\x}_i) - f_i (\x_i^*)}{\delta_i (\epsilon)}
    = \frac{\epsilon}{ \delta_i (\epsilon)}.  
    \label{eqn: lb slope}
\end{equation}
Therefore, combining \eqref{eqn: cos theta} and \eqref{eqn: lb slope}, we obtain
\begin{equation}
    \cos \hat{\theta}_i (\x) 
    \geq  \frac{\epsilon}{L \delta_i (\epsilon)}.  
    \label{eqn: cos theta hat}
\end{equation}
However, from Assumption~\ref{asm: convex}, we have 
\begin{equation*}
    f_i (\x_i^*) \geq f_i ( \Tilde{\x}_i ) + \langle \Tilde{\g}_i ( \Tilde{\x}_i ), \; \x_i^* - \Tilde{\x}_i \rangle
\end{equation*}
where $\Tilde{\g}_i ( \Tilde{\x}_i ) \in \partial f_i ( \Tilde{\x}_i )$.
Since $\| \Tilde{\g}_i ( \Tilde{\x}_i ) \| \leq L$ by Assumption~\ref{asm: gradient_bound} and $\| \x_i^* - \Tilde{\x}_i \| \leq \delta_i (\epsilon)$, we get 
\begin{equation*}
    \epsilon =  f_i ( \Tilde{\x}_i ) - f_i ( \x_i^* ) 
    \leq - \langle \Tilde{\g}_i ( \Tilde{\x}_i ), \; \x_i^* - \Tilde{\x}_i \rangle 
    \leq L \delta_i (\epsilon).
\end{equation*}
From $\epsilon \in \R_{>0}$ and the above inequality, the inequality \eqref{eqn: cos theta hat} becomes
\begin{equation*}
    \hat{\theta}_i (\x) 
    \leq \arccos \Big( \frac{\epsilon}{L \delta_i (\epsilon)} \Big) 
    := \theta_i (\epsilon) < \frac{\pi}{2},
\end{equation*}
which completes the proof.
\end{proof}

%% file: contents/sec-supp_proof_true.tex
\section{Proof of Theorem~\ref{thm: true_sol}}
\label{sec: proof_true_sol}

\begin{proof}
We will show that the summation of any subgradients of the regular nodes' functions at any point outside the region $\B \big( \y [\infty], \; \inf_{\epsilon > 0} s^*( 0, \epsilon) \big)$ cannot be zero. 

Let $\x_0$ be a point outside $\B \big( \y [\infty], \; \inf_{\epsilon > 0} s^*( 0, \epsilon) \big)$. 
Since $\| \x_0 - \y [\infty] \| > \max_{v_i \in \mathcal{R}} \{ \Tilde{R}_i + \delta_i(\epsilon) \}$ for some $\epsilon > 0$, we have that $\x_0 \notin \mathcal{C}_i(\epsilon)$ for all $v_i \in \mathcal{R}$.
By the definition of $\mathcal{C}_i(\epsilon)$ in \eqref{def: sublevel_set}, we have 
$f_i (\x_0) > f_i (\x_i^*) + \epsilon$ for all $v_i \in \mathcal{R}$. 
Since the functions $f_i$ are convex, we obtain $\g_i (\x_0) \neq \mathbf{0}$ for all $v_i \in \mathcal{R}$ where $\g_i (\x_0) \in \partial f_i (\x_0)$.

Consider the angle between the vectors $\x_0 - \x_i^*$ and $\x_0 - \y [\infty]$. 
If $\Tilde{R}_i = 0$, from \eqref{def: tilde R_i}, we have $\x_i^* = \y [\infty]$ which implies that $\angle (\x_0 - \x_i^*, \; \x_0 - \y [\infty]) = 0$.
Suppose $\Tilde{R}_i > 0$. Using Lemma~\ref{lem: max_angle}, we can bound the angle as follows: 
\begin{equation*}
    \angle (\x_0 - \x_i^*, \; \x_0 - \y [\infty])
    \leq \arcsin \Big( \frac{ \Tilde{R}_i }{ \| \x_0 - \y [\infty] \|} \Big). 
\end{equation*}
Since $\| \x_0 - \y [\infty] \| > \max_{v_i \in \mathcal{R}} \{ \Tilde{R}_i \sec \theta_i (\epsilon) \}$ for some $\epsilon > 0$ and $\arcsin(x)$ is an increasing function in $x \in [-1, 1]$, we have
\begin{equation*}
    \angle (\x_0 - \x_i^*, \; \x_0 - \y [\infty])
    < \arcsin \Big( \frac{ \Tilde{R}_i }{ \Tilde{R}_i \sec \theta_i (\epsilon) } \Big) 
\end{equation*}
and that $\arcsin ( \cos \theta_i (\epsilon) ) = \frac{\pi}{2} - \theta_i (\epsilon).$ 
Using Proposition \ref{prop: grad angle} and the inequality above, we can bound the angle between the vectors $\g_i (\x_0)$ and $\x_0 - \y [\infty]$ as follows:
\begin{equation*}
    \angle ( \g_i (\x_0), \; \x_0 - \y [\infty] )
    \leq \angle ( \g_i (\x_0), \; \x_0 - \x_i^*) + \angle (\x_0 - \x_i^*, \; \x_0 - \y [\infty] )
    < \theta_i (\epsilon) + \Big( \frac{\pi}{2} - \theta_i (\epsilon) \Big) = \frac{\pi}{2}.
\end{equation*}
Note that the first inequality is obtained from \cite[Corollary~12]{castano2016angles}. Let $\u = \frac{\x_0 - \y [\infty]}{ \| \x_0 - \y [\infty] \| }$. Compute the inner product
\begin{equation*}
    \Big\langle \sum_{v_i \in \mathcal{R}} \g_i (\x_0), \; \u \Big\rangle
    = \sum_{v_i \in \mathcal{R}} \| \g_i (\x_0) \| \cos \angle ( \g_i (\x_0), \; \x_0 - \y [\infty] ). 
\end{equation*}
The RHS of the above equation is strictly greater than zero
since $\| \g_i (\x_0) \| > 0$ and $\cos \angle ( \g_i (\x_0), \; \x_0 - \y [\infty] ) > 0$ for all $v_i \in \mathcal{R}$. This implies that $\sum_{v_i \in \mathcal{R}} \g_i (\x_0) \neq \mathbf{0}$. 
Since we can arbitrarily choose $\g_i (\x_0)$ from the set $\partial f_i (\x_0)$, we have $\mathbf{0} \notin \partial f (\x_0)$ where $f(\x) = \frac{1}{| \mathcal{R} |} \sum_{v_i \in \mathcal{R}} f_i (\x)$. 
\end{proof}

%% file: contents/sec-supp_proof_conv.tex
\section{Proof of Results in Section~\ref{subsec: grad_update analysis}} 
\label{sec: proof_main_1}

\begin{proof}[Proof of Lemma~\ref{lem: dist grad update 1}]
From Proposition~\ref{prop: aux convergence}, the limit point $\y [\infty] \in \R^d$ exists.
Consider a time-step $k \in \N$ such that $k \geq k_1^*$. Using the gradient step \eqref{eqn: grad dynamic}, we can write
\begin{equation*}
    \| \x_i [k+1] - \y[\infty] \|
    = \| \z_i[k] - \y[\infty] - \eta[k] \; \g_i[k] \|
    \leq \| \z_i[k] - \y[\infty] \| + \eta[k] \; \| \g_i[k] \|  .
\end{equation*}
Using Assumptions \ref{asm: gradient_bound} and \ref{asm: step-size}, and $\| \z_i[k] - \y[\infty] \| \leq \max_{v_j \in \mathcal{R}} \{ \Tilde{R}_j + \delta_j \}$, we obtain
\begin{equation*}
    \| \x_i [k+1] - \y[\infty] \|
    \leq \max_{v_j \in \mathcal{R}} \{ \Tilde{R}_j + \delta_j \}  + \eta[k_1^*] L.
\end{equation*}
By the definition of $k_1^*$ and $s^*$ in \eqref{def: conv radius}, the above inequality becomes
\begin{equation*}
    \| \x_i [k+1] - \y[\infty] \|
    \leq \max_{v_j \in \mathcal{R}} \{ \Tilde{R}_j + \delta_j \}  + \xi 
    \leq s^*,
\end{equation*}
which completes the proof.
\end{proof}

\begin{proof}[Proof of Lemma~\ref{lem: grad update}]
From Proposition~\ref{prop: aux convergence}, the limit point $\y [\infty] \in \R^d$ exists. 
Consider an agent $v_i \in \mathcal{R}$ and a time-step $k \in \N$ for which the condition in the lemma holds. 
Since $\x_i^* \in \B (\y [\infty], \Tilde{R}_i)$ from \eqref{def: tilde R_i}, we have 
\begin{equation}
    \angle (\x_i^* - \z_i[k], \; \y [\infty] - \z_i[k])
    \leq \max_{\u \in \B(\y [\infty], \Tilde{R}_i)}  \angle (\u - \z_i[k], \; \y [\infty] - \z_i[k])
    = \arcsin \frac{\Tilde{R}_i}{\| \z_i[k] - \y [\infty] \|}, \label{def: angle_phi}
\end{equation}
where the last step is from using Lemma~\ref{lem: max_angle}.
Using the gradient step \eqref{eqn: grad dynamic}, we can write $\angle ( \x_i[k+1] - \z_i[k], \; \y [\infty] - \z_i[k]) = \angle ( - \eta[k] \g_i[k], \; \y [\infty] - \z_i[k] )$.
Since for all $v_i \in \mathcal{R}$ and $k \in \N$,
\begin{equation*}
    \angle ( - \eta[k] \g_i[k], \; \y [\infty] - \z_i[k] )
    \leq \angle ( - \eta[k] \g_i[k], \; \x_i^* - \z_i[k]  )
    + \angle ( \x_i^* - \z_i[k], \; \y [\infty] - \z_i[k] )
\end{equation*}
by \cite[Corollary~12]{castano2016angles}, applying Proposition \ref{prop: grad angle} and inequality \eqref{def: angle_phi}, we have
\begin{equation}
    \angle ( \x_i[k+1] - \z_i[k], \; \y [\infty] - \z_i[k])
    \leq \theta_i + \arcsin \frac{\Tilde{R}_i}{\| \z_i[k] - \y [\infty] \|} 
    := \psi_i[k].  
    \label{def: psi angle}
\end{equation}
Note that $\psi_i[k] \in [0, \pi)$ since $\theta_i \in \big[ 0, \frac{\pi}{2} \big)$ and $\arcsin \frac{\Tilde{R}_i}{\| \z_i[k] - \y [\infty] \|} \in \big[ 0, \frac{\pi}{2} \big]$.
Then, consider the triangle which has the vertices at $\x_i[k+1]$, $\z_i[k]$, and $\y[\infty]$.
We can calculate the square of the distance by using the law of cosines:
\begin{multline*}
    \| \x_i[k+1] - \y[\infty] \|^2 
    = \; \| \x_i[k+1] - \z_i[k] \|^2 + \| \y[\infty] - \z_i[k] \|^2 \\
    - 2 \| \x_i[k+1] - \z_i[k] \| \cdot \| \y[\infty] - \z_i[k] \| 
    \cos \angle (\x_i[k+1] - \z_i[k], \; \y[\infty] - \z_i[k]  )  .
\end{multline*}    
Using the gradient step \eqref{eqn: grad dynamic} and the inequality \eqref{def: psi angle}, we get 
\begin{equation}    
    \| \x_i[k+1] - \y[\infty] \|^2 
    \leq   \eta^2[k] \; \| \g_i[k] \|^2 + \| \z_i[k] - \y [\infty] \|^2 
    - 2 \; \eta[k] \; \| \g_i[k] \| \cdot \| \z_i[k] - \y [\infty] \| \cos \psi_i[k].  
    \label{eqn: distant bound}
\end{equation}
In addition, we can simplify the term $\| \z_i[k] - \y [\infty] \| \cos \psi_i[k]$ in the above inequality using the definition of $\psi_i[k]$ in \eqref{def: psi angle}. Regarding this, we can write
\begin{equation*}
    \| \z_i[k] - \y [\infty] \| \cos \psi_i[k] 
    = \sqrt{ \| \z_i[k] - \y [\infty] \|^2 - \Tilde{R}_i^2 } \cdot \cos \theta_i - \Tilde{R}_i \sin \theta_i .
\end{equation*}
Substituting the above equation into \eqref{eqn: distant bound}, we obtain the result.
\end{proof}

\begin{proof}[Proof of Lemma~\ref{lem: dist grad update 2}]
First, note that by the definition of $s^*$ in \eqref{def: conv radius}, we have $a_i^+ > 0$ and $a_i^- < 0$ since $ s^*  \geq \Tilde{R}_i + \xi$,
and $b_i > 0$ since $s^* \geq \Tilde{R}_i \sec \theta_i + \xi$.

For $v_i \in \mathcal{R}$, let $\Gamma_i: [ \Tilde{R}_i, \infty) \times \R_+ \to \R$ be the function
\begin{equation}
    \Gamma_i( p, l ) := p^2 - \Delta_i( p, l ),   \label{def: Gamma function}
\end{equation}
where function $\Delta_i$ is defined in \eqref{def: delta}.
Consider an agent $v_i \in \mathcal{R}$ and a time-step $k \in \N$ such that $k \geq k^*_2$.
We can compute the second derivative of $\Gamma_i( p, l )$ with respect to $p$ as follows:
\begin{equation*}
    \frac{\partial^2 \Gamma_i}{\partial p^2} = 2 + 2 l \Tilde{R}_i^2 (p^2 - \Tilde{R}_i^2)^{- \frac{3}{2}} \cos \theta_i .
\end{equation*}
Note that $\frac{\partial^2 \Gamma_i}{\partial p^2} > 0$ for all $p \in (\Tilde{R}_i, \infty)$. This implies that 
\begin{equation}
    \sup_{ p \in ( \max_{v_j \in \mathcal{R}} \{ \Tilde{R}_j + \delta_j \}, \; s^* ] } \Gamma_i (p, l) 
    \leq \max_{ p \in [ \Tilde{R}_i, \; s^* ] } \Gamma_i (p, l) 
    = \max \big\{ \Gamma_i( \Tilde{R}_i, l ), \; \Gamma_i ( s^*, l ) \big\}.  
    \label{eqn: gamma chain}
\end{equation} 
First, let consider $\Gamma_i( \Tilde{R}_i, \; l )$.
From the definition of $\Gamma_i$ in \eqref{def: Gamma function}, we have that
\begin{align}
    \Gamma_i ( \Tilde{R}_i, l ) \leq ( s^* )^2 \quad \Longleftrightarrow \quad l \in [ a_i^-, \; a_i^+ ],   
    \label{eqn: gamma_leq_1}
\end{align}
where $a_i^+$ and $a_i^-$ are defined in \eqref{def: a_i & b_i}.
Using Assumption \ref{asm: gradient_bound} and \ref{asm: step-size}, and the definition of $k_2^*$, we have that 
\begin{equation*}
    \eta[k] \; \| \g_i[k] \| \leq \eta[k] L \leq \min_{v_j \in \mathcal{R}} \big\{ \min \{ a_j^+, \; b_j \} \big\} \leq  a_i^+ .
\end{equation*}
By the above inequality and statement \eqref{eqn: gamma_leq_1}, we obtain that 
\begin{equation}
     \Gamma_i ( \Tilde{R}_i, \; \eta[k] \; \| \g_i[k] \| ) \leq ( s^* )^2. 
     \label{eqn: gamma ineq1}
\end{equation}
Now, let consider $\Gamma_i ( s^*, \; l )$. 
From the definition of $\Gamma_i$ in \eqref{def: Gamma function}, we have that
\begin{equation}
     \Gamma_i ( s^*, l ) \leq ( s^* )^2  \quad \Longleftrightarrow \quad
     l \in [ 0, \; b_i ],  
     \label{eqn: gamma_leq_2}
\end{equation}
where $b_i$ is defined in \eqref{def: a_i & b_i}.
Using Assumption \ref{asm: gradient_bound} and \ref{asm: step-size}, and the definition of $k_2^*$, we have that 
\begin{equation*}
    \eta[k] \; \| \g_i[k] \| \leq \eta[k] L \leq \min_{v_j \in \mathcal{R}} \big\{ \min \{ a_j^+, \; b_j \} \big\} \leq  b_i .
\end{equation*}
By the above inequality and statement \eqref{eqn: gamma_leq_2}, we obtain that
\begin{equation}
     \Gamma_i ( s^*, \; \eta[k] \; \| \g_i[k] \| ) \leq ( s^* )^2.  
     \label{eqn: gamma ineq2}
\end{equation}
Combine \eqref{eqn: gamma ineq1} and \eqref{eqn: gamma ineq2} to get that 
\begin{equation}
    \max \big\{ \Gamma_i ( \Tilde{R}_i, \; \eta[k] \; \| \g_i[k] \| ), \; \Gamma_i ( s^*(\xi), \; \eta[k] \; \| \g_i[k] \| ) \big\} 
    \leq ( s^* )^2. 
    \label{eqn: max bound s}
\end{equation}
From Lemma \ref{lem: grad update}, we can write 
\begin{equation*}
    \| \x_i[k+1] - \y[\infty] \|^2 
    \leq \Gamma_i( \| \z_i[k] - \y[\infty] \|, \; \eta[k] \; \| \g_i[k] \| ).
\end{equation*}
Applying \eqref{eqn: gamma chain} and \eqref{eqn: max bound s}, respectively to the above inequality yields the result.
\end{proof}

\begin{proof}[Proof of Lemma~\ref{lem: Delta lower bound}]
Consider any time-step $k \geq k^*_3$ and agent $v_i \in \mathcal{I}_z[k]$.
By the definition of the function $\Delta_i$ in \eqref{def: delta}, it is clear that if $p_1 > p_2 \geq \Tilde{R}_i$ then $\Delta_i( p_1, l) > \Delta_i( p_2, l)$. Then, we get 
\begin{equation}
    \Delta_i( \| \z_i[k] - \y[\infty] \|, \eta[k] \; \| \g_i[k] \| ) > \Delta_i (  s^*, \eta[k] \; \| \g_i[k] \| ).  
    \label{eqn: delta_ineq}
\end{equation}
Furthermore, the function $\Delta_i$ satisfies
\begin{equation}
    \Delta_i ( p, l) \geq \Big( \sqrt{ p^2 - \Tilde{R}_i^2 } \cos \theta_i - \Tilde{R}_i \sin \theta_i \Big) l 
    \quad \Longleftrightarrow \quad 
    l \in \Big[ 0, \; \sqrt{ p^2 - \Tilde{R}_i^2 } \cos \theta_i - \Tilde{R}_i \sin \theta_i  \Big].  
    \label{eqn: l_range}
\end{equation}

We restate inequality \eqref{eqn: lb slope} obtained in the proof of Proposition~\ref{prop: grad angle} here:
\begin{align}
    \frac{f_i(\x) - f_i(\x_i^*)}{\| \x - \x_i^* \|}
    \geq \frac{\epsilon}{\delta_i(\epsilon)}.     
    \label{eqn: lb_grad}
\end{align}
Recall the definition of $\mathcal{C}_i(\epsilon)$ in \eqref{def: sublevel_set}.
For $\x \notin \mathcal{C}_i(\epsilon)$, from the definition of convex functions, we have
$- \langle g_i(\x), \; \x_i^* - \x \rangle \geq f_i(\x) - f_i(\x_i^*)$.  
Using the inequality \eqref{eqn: lb_grad}, we obtain
\begin{align*}
    \| \g_i (\x) \| \geq
    \frac{f_i(\x) - f_i(\x_i^*)}{\| \x - \x_i^* \|}
    \geq \frac{\epsilon}{\delta_i(\epsilon)} = \underline{L}_i.   
\end{align*}
Using the above inequality, Assumption \ref{asm: gradient_bound}, and the definition of $k_3^*$, we have that
\begin{align}
    \eta[k] \underline{L}_i \leq \eta[k] \; \| \g_i[k] \| \leq \eta[k] L \leq \frac{b_i}{2}.  
    \label{eqn: lb_ub}
\end{align}
Since $\eta[k] \; \| \g_i[k] \| \in [ 0, \; \frac{b_i}{2}  ]$, we can apply \eqref{eqn: l_range} to get
\begin{align}
    \Delta_i \big(  s^*, \; \eta[k] \; \| \g_i[k] \| \big) \geq \frac{b_i}{2} \eta[k] \; \| \g_i[k] \|.  
    \label{eqn: delta_geq}
\end{align}
Combine \eqref{eqn: delta_ineq}, \eqref{eqn: delta_geq}, and the first inequality of \eqref{eqn: lb_ub} to obtain the result.
\end{proof}

\section{Proof of Results in Section~\ref{subsec: bounds}} 
\label{sec: proof_main_2}

\begin{proof}[Proof of Proposition~\ref{prop: aux inexact}]
From Proposition~\ref{prop: aux convergence}, the limit point $\y [\infty] \in \R^d$ exists.
For all $v_i$, $v_j \in \mathcal{R}$, we have
\begin{equation*}
    \| \x_j[k] - \y_i[k] \| 
    \leq \| \x_j[k] - \y[\infty] \| + \| \y_i[k] - \y[\infty] \|. 
\end{equation*}
Apply Lemma~\ref{lem: dist to curr aux} to obtain that for all $v_i \in \mathcal{R}$, we have
\begin{equation*}
    \| \z_i[k] - \y_i[k] \| 
    \leq \max_{v_j \in \mathcal{R}} \| \x_j[k] - \y[\infty] \| + \| \y_i[k] - \y[\infty] \|.
\end{equation*}
Substituting the above inequality into 
\begin{equation*}
    \| \z_i[k] - \y[\infty] \| 
    \leq \| \z_i[k] - \y_i[k] \| + \| \y_i[k] - \y[\infty] \|,
\end{equation*}
we obtain the result.
\end{proof}

\begin{proof}[Proof of Lemma~\ref{lem: two bounds}]
Suppose $\max_{v_i \in \mathcal{R}} \| \x_i[k] - \y[\infty] \| \leq \phi[k]$ for a time-step $k \geq k_0$.
From Proposition~\ref{prop: aux convergence} and \ref{prop: aux inexact}, we have that for all $v_i \in \mathcal{R}$,
\begin{equation}
    \| \z_i[k] - \y[\infty] \| \leq \phi[k] + 2 \beta e ^{- \alpha k},  
    \label{eqn: bound z_i[k]}
\end{equation}
where $\alpha$ and $\beta$ are defined in Proposition~\ref{prop: aux convergence}.

Recall the definition of $\mathcal{I}_z[k]$ from \eqref{def: I outside s}. For all $v_i \in \mathcal{I}_z[k]$, from Lemma~\ref{lem: grad update} and \ref{lem: Delta lower bound}, we have 
\begin{equation*}
    \| \x_i[k+1] - \y[\infty] \|^2 \leq \| \z_i[k] - \y[\infty] \|^2 - \frac{1}{2} b_i \underline{L}_i \eta[k].  
\end{equation*}
Applying \eqref{eqn: bound z_i[k]} to the above inequality, we obtain that for all $v_i \in \mathcal{I}_z[k]$, 
\begin{equation*}
    \| \x_i[k+1] - \y[\infty] \|^2 
    \leq \big( \phi[k] + 2 \beta e ^{- \alpha k} \big)^2 - \frac{1}{2} b_i \underline{L}_i \eta[k] 
    \leq \big( \phi[k] + 2 \beta e ^{- \alpha k} \big)^2 - \frac{1}{2} \eta[k] \min_{v_j \in \mathcal{R}} b_j \underline{L}_j .
\end{equation*}
On the other hand, for all $v_i \in \mathcal{R} \setminus \mathcal{I}_z[k]$, we have $\| \z_i[k] - \y[\infty] \| \leq s^*$ by the definition of $\mathcal{I}_z[k]$. From Lemma \ref{lem: dist grad update 1} and \ref{lem: dist grad update 2}, we get $\| \x_i[k+1] - \y[\infty] \| \leq s^*$ for all $v_i \in \mathcal{R} \setminus \mathcal{I}_z[k]$. Therefore, we conclude that for all $v_i \in \mathcal{R}$, 
\begin{equation*}
    \| \x_i[k+1] - \y[\infty] \|^2 \leq \max \Big\{ ( s^* )^2, 
    \big( \phi[k] + 2 \beta e ^{- \alpha k} \big)^2 - \frac{1}{2} \eta[k] \min_{v_j \in \mathcal{R}} b_j \underline{L}_j \Big\}.  
\end{equation*}
Using the update rule \eqref{eqn: phi k+1}, the above inequality implies that
$\max_{v_i \in \mathcal{R}} \| \x_i[k+1] - \y[\infty] \| \leq \phi[k+1]$.

Next, consider a time-step $k \in \N$. From the gradient update step \eqref{eqn: grad dynamic}, for all $v_i \in \mathcal{R}$, we have
\begin{equation*}
    \| \x_i[k+1] - \y[\infty] \| = \| \z_i[k] - \y[\infty] - \eta[k] \g_i[k] \| 
    \leq  \| \z_i[k] - \y[\infty] \| + \eta[k] \; \| \g_i[k] \|.
\end{equation*}
Since $\| \g_i[k] \| \leq L$ from Assumption~\ref{asm: gradient_bound}, we can rewrite the above inequality as
\begin{equation}
    \| \x_i[k+1] - \y[\infty] \| \leq  \| \z_i[k] - \y[\infty] \| + \eta[k] \; L. 
    \label{eqn: norm grad update}
\end{equation}
On the other hand, from Proposition~\ref{prop: aux convergence} and \ref{prop: aux inexact}, for all $v_i \in \mathcal{R}$, we have
\begin{equation}
    \| \z_i[k] - \y[\infty] \| 
    \leq \max_{v_j \in \mathcal{R}} \| \x_j[k] - \y[\infty] \|  + 2 \beta e ^{- \alpha k}. 
    \label{eqn: inexact update}
\end{equation}
Combine the inequalities \eqref{eqn: norm grad update} and \eqref{eqn: inexact update} together and apply the result recursively to obtain
\begin{equation*}
    \max_{v_i \in \mathcal{R}} \| \x_i[k_0] - \y[\infty] \| 
    \leq \max_{v_i \in \mathcal{R}} \| \x_i[0] - \y[\infty] \|  
    + 2 \beta \sum_{k=0}^{k_0-1} e^{-\alpha k} + L \sum_{k=0}^{k_0-1} \eta [k].
\end{equation*}
Since the RHS of the above inequality is $\phi [k_0]$, this completes the first part of the proof.

Consider a time-step $k \in \N$ such that $k \geq k_0$. From the update equation \eqref{eqn: phi k+1}, using the fact that $\frac{1}{2} \eta[k] \min_{v_i \in \mathcal{R}} b_i \underline{L}_i > 0$ for all $k \in \N$, we can write 
\begin{equation*}
    \phi [k+1] <  \max \big\{ s^*, \; \phi[k] \big\}  + 2 \beta e^{- \alpha k}.
\end{equation*}
Applying the above inequality recursively, we can write that for all $k \geq k_0$,
\begin{equation*}
    \phi[k] < \max \big\{ s^*, \; \phi[k_0] \big\}  + 2 \beta \sum_{k' = k_0}^{k-1} e^{- \alpha k'}.
\end{equation*}
Substituting equation \eqref{eqn: phi[k0]} into the above inequality and using the fact that $\sum_{k'=0}^{k-1} e^{- \alpha k'} < \sum_{k'=0}^{\infty} e^{- \alpha k'} = \frac{1}{1 - e^{-\alpha}}$ for all $k \in \N$, we obtain the uniform bound as follows:
\begin{equation*}
    \phi[k] < \max \Big\{ s^*, \; \max_{v_i \in \mathcal{R}} \| \x_i[0] - \y[\infty] \|  
    + L \sum_{k'=0}^{k_0-1} \eta[k'] \Big\} 
    + \frac{2 \beta}{1 - e^{-\alpha}}.
    \label{eqn: phi bar}
\end{equation*}
Setting the RHS of the above inequality to $\bar{\phi}$, we obtain the result.
\end{proof}

\section{Proof of Results in Section~\ref{subsec: conv analysis}} 
\label{sec: proof_main_3}

Lemma~\ref{lem: update imp} is used to establish the proof of Proposition~\ref{prop: finite time}.

\begin{proof}[Proof of Lemma~\ref{lem: update imp}]
Suppose that there exists a sequence $\{ u[k] \}_{k=0}^{\infty} \subset \R_{\geq 0}$ that satisfies the given update rule.
Since $\hat{\eta}[k] \in \R_{\geq 0}$ for all $k \in \N$, we have $u^2[k+1] \leq (u[k] + \gamma_1 \lambda^k)^2$.
Since $u[k] \geq 0$ for all $k \in \N$, it follows that
\begin{equation*}
    0 \leq u[k+1] \leq | u[k] + \gamma_1 \lambda^k | \leq u[k] + \gamma_1 | \lambda |^k.
\end{equation*}
Apply the above inequality recursively to obtain that for all $k \in \N$,
\begin{equation*}
    u[k] \leq u[0] + \gamma_1 \sum_{\ell = 0}^{k} | \lambda |^\ell
    \leq u[0] + \frac{\gamma_1}{1- |\lambda|} := \bar{u}.
\end{equation*}
From the update rule, we can write
\begin{equation*}
    u^2[k+1] = u^2[k] + 2 \gamma_1 \lambda^k u[k] 
    + \gamma_1^2 \lambda^{2k} - \gamma_2 \hat{\eta}[k] 
    \leq u^2[k] + 2 \gamma_1 \lambda^k \bar{u} + \gamma_1^2 \lambda^{2k} - \gamma_2 \hat{\eta}[k].
\end{equation*}
Applying the above inequality recursively, we obtain
\begin{equation*}
    u^2[k] \leq u^2[0] + 2 \gamma_1 \bar{u} \sum_{\ell=0}^{k} \lambda^{\ell} + \gamma_1^2 \sum_{\ell=0}^{k} \lambda^{2 \ell} - \gamma_2 \sum_{\ell=0}^{k} \hat{\eta}[\ell].
\end{equation*}
However, the first three terms on the RHS are bounded in $k$ while the last term is unbounded. This implies that there exists a time-step $\tilde{k} \in \N$ such that $u^2[\tilde{k}] < 0$ which contradicts the fact that $u[\tilde{k}] \in \R_{\geq 0}$.
\end{proof}

\begin{proof}[Proof of Proposition~\ref{prop: finite time}]
Let $\nu = \frac{1}{2} \min_{v_i \in \mathcal{R}} b_i \underline{L}_i$ to simplify the notations. 
Let $k^*_4 \in \N$ be a time-step such that 
$\eta[k^*_4] \geq \frac{4 \beta}{\nu} \big( \bar{\phi} e^{- \alpha k^*_4} + \beta e^{- 2 \alpha k^*_4} \big)$. Note that $k^*_4 \in \N$ exists since $\eta[k]$ decreases slower than the exponential decay due to its form given in Assumption~\ref{asm: step-size}.

First, we will show that if the time-step $k \in \N$ satisfies $k \geq \max \{ k_0, k_4^* \}$ and $\big( \phi[k] + 2 \beta e ^{- \alpha k} \big)^2 - \nu \eta[k] > ( s^* )^2$, then
\begin{equation}
    \phi[k+1] < \phi[k]. 
    \label{eqn: decreasing phi}
\end{equation}

Consider a time-step $k \geq \max \{ k_0, k_4^* \}$. Since $\big( \phi[k] + 2 \beta e ^{- \alpha k} \big)^2 - \nu \eta[k]  > ( s^* )^2$, the update equation \eqref{eqn: phi k+1} reduces to
\begin{equation}
    \phi^2[k+1] = \big( \phi[k] + 2 \beta e ^{- \alpha k} \big)^2 - \nu \eta[k].
    \label{eqn: redeuced phi update}
\end{equation}
Using the definition of $k_4^*$, from $k \geq k_4^*$, we can write $\eta[k] \geq \frac{4 \beta}{\nu} \big( \bar{\phi} e^{- \alpha k} + \beta e^{- 2 \alpha k} \big)$. Since $\phi [k] < \bar{\phi}$, we have that
\begin{equation*}
    \eta[k] > \frac{4 \beta}{\nu} \big( \phi [k] e^{- \alpha k} + \beta e^{- 2 \alpha k} \big).
\end{equation*}
By multiplying $\nu$ and adding $\phi^2 [k]$ to both sides, and then rearranging, we can write 
$\phi^2 [k] > \big( \phi[k] + 2 \beta e^{- \alpha k} \big)^2 - \nu \eta[k]$,
which is equivalent to $\phi[k+1] < \phi[k]$ by \eqref{eqn: redeuced phi update}. This completes our claim.

Next, we will show that there exists a time-step $\tilde{K} \in \N$ such that 
$\phi \big[ \max \{ k_0, k_4^* \} + \tilde{K} \big] = s^*$.

Suppose that $\big( \phi[k] + 2 \beta e ^{- \alpha k} \big)^2 - \nu \eta[k] > ( s^* )^2$ for all $k \geq \max \{ k_0, k_4^* \}$. Then, the update equation \eqref{eqn: phi k+1} reduces to
\begin{equation*}
    \phi^2[k+1] = \big( \phi[k] + 2 \beta e ^{- \alpha k} \big)^2 - \nu \eta[k].
\end{equation*}
However, since $\phi[k]$ is non-negative for all $k \in \N$ by its definition, from Lemma~\ref{lem: update imp}, there is no sequence $\{ \phi[k] \}_{k=k_0}^{\infty}$ that can satisfy the above update rule. 
Hence, there exists a constant $\tilde{K} \in \N$ such that
\begin{equation*}
    \big( \phi[ k'] + 2 \beta e ^{- \alpha k'} \big)^2 - \nu \eta[ k'] \leq (s^*)^2,
\end{equation*}
where $k' = \max \{ k_0, k_4^* \} + \tilde{K} - 1$,
which yields $\phi \big[ \max \{ k_0, k_4^* \} + \tilde{K} \big] = s^*$ by the equation \eqref{eqn: phi k+1}. This completes the second claim.

Consider any time-step $k \in \N$ such that $k \geq \max \{ k_0, k_4^* \} + \tilde{K}$ and $\phi [k] = s^*$. Such a time-step exists due to the argument above. Then, suppose $\big( \phi[k] + 2 \beta e ^{- \alpha k} \big)^2 - \nu \eta[k]  > ( s^* )^2$. From \eqref{eqn: decreasing phi}, we have that $\phi [k+1] < s^*$ which is not possible due to the fact that $\phi [k'] \geq s^*$ for all $k' \geq k_0$ from the update equation \eqref{eqn: phi k+1}. Hence, we conclude that $\big( \phi[k] + 2 \beta e ^{- \alpha k} \big)^2 - \nu \eta[k] \leq ( s^* )^2$, and $\phi [k+1] = s^*$ by \eqref{eqn: phi k+1}. 
This means that $\phi [k] = s^*$ for all $k \geq \max \{ k_0, k_4^* \} + \tilde{K}$. Then, by the definition of $\phi[k]$, we can rewrite the equation as
$\max_{v_i \in \mathcal{R}} \| \x_i[ k ] - \y[\infty] \| \leq s^*$
for all $k \geq \max \{ k_0, k_4^* \} + \tilde{K} := K$ which completes the proof.
\end{proof}

Finally, we utilize the finite-time convergence result from Proposition~\ref{prop: finite time} to give a proof for Theorem \ref{thm: convergence} presented below.

\begin{proof}[Proof of Theorem \ref{thm: convergence}]
From Proposition~\ref{prop: finite time}, for a fixed $\xi \in \R_{>0}$ and $\epsilon \in \R_{>0}$, we have that for all $k \geq K$,
\begin{equation*}
    \max_{v_i \in \mathcal{R}} \| \x_i[k] - \y[\infty] \| \leq s^*(\xi, \epsilon).
\end{equation*}
Note that $K$ is a function of $\xi$ and $\epsilon$.
However, the above inequality implies that
$\limsup_k \max_{v_i \in \mathcal{R}} \| \x_i[k] - \y[\infty] \| \leq s^*(\xi, \epsilon)$.
Since the inequality is valid for all $\xi > 0$ and $\epsilon > 0$, and 
$\inf_{\xi > 0, \; \epsilon > 0} s^*(\xi, \epsilon) = \inf_{\epsilon > 0} s^*(0, \epsilon)$
by the definition of $s^*(\xi, \epsilon)$ in \eqref{def: conv radius}, we have
\begin{equation*}
    \limsup_k \max_{v_i \in \mathcal{R}} \| \x_i[k] - \y[\infty] \| \leq \inf_{\epsilon > 0} s^*(0, \epsilon),
\end{equation*}
which completes the proof.
\end{proof}